\documentclass[12pt,authoryear,review]{elsarticle}

\makeatletter
\def\ps@pprintTitle{%
 \let\@oddhead\@empty
 \let\@evenhead\@empty
 \def\@oddfoot{}%
 \let\@evenfoot\@oddfoot}
\makeatother

\renewcommand{\baselinestretch}{1.5}

\usepackage{titlesec}
\usepackage{color}
\usepackage[table]{xcolor}
\usepackage{booktabs}
\usepackage{amsmath,amssymb,amsfonts}
\usepackage{mathrsfs}
\usepackage{graphics}
\usepackage{graphicx}
\usepackage{dsfont}
\usepackage{subfigure}
\usepackage{parskip}
\usepackage{url}
\usepackage[font=footnotesize,format=hang,width=0.8\textwidth]{caption}
\usepackage[titletoc,title]{appendix}

\usepackage{enumitem}
\setitemize{itemsep=0pt}

\usepackage{soul}

\numberwithin{equation}{section}
\newtheorem{Theorem}{Theorem}

\newtheorem{Proposition}{Proposition}
\newtheorem{Corollary}{Corollary}
\newtheorem{Remark}{Remark}
\newtheorem{Assumption}{Assumption}
\newtheorem{Definition}{Definition}
\newenvironment{proof}[1][Proof]{\textbf{#1.} }{\ \rule{0.5em}{0.5em}}

\newcommand{\EE}{{\mathbb E}}
\newcommand{\VV}{{\mathbb V}}

\newcommand{\RR}{{\mathbb R}}
\newcommand{\PP}{{\mathbb P}}
\newcommand{\LL}{{\mathcal L}}

\newcommand{\Ft}{{\mathcal F}}

\newcommand{\As}{{\mathcal A}}

\newcommand{\ones}{\mathbf{1}}

\newcommand{\vX}{{\boldsymbol X}}
\newcommand{\vW}{{\boldsymbol W}}
\newcommand{\vgamma}{{\boldsymbol \gamma}}
\newcommand{\vdelta}{{\boldsymbol \delta}}
\newcommand{\vxi}{{\boldsymbol \xi}}
\newcommand{\vpi}{{\boldsymbol \pi}}
\newcommand{\valpha}{{\boldsymbol \alpha}}
\newcommand{\vmu}{{\boldsymbol \mu}}
\newcommand{\veta}{{\boldsymbol \eta}}
\newcommand{\vzeta}{{\boldsymbol \zeta}}
\newcommand{\vrho}{{\boldsymbol \rho}}
\newcommand{\mA}{{\mathbf A}}
\newcommand{\mI}{{\mathbf I}}
\newcommand{\vB}{{\boldsymbol B}}
\newcommand{\vd}{{\boldsymbol d}}
\newcommand{\mQ}{{\mathbf Q}}

\newcommand{\vx}{{\boldsymbol x}}
\newcommand{\mSigma}{\mathbf \Sigma}

\usepackage{color}

\begin{document}

\begin{frontmatter}

\title{\textbf{Outperformance and Tracking}: \\
\textbf{\normalsize Dynamic Asset Allocation for Active and Passive Portfolio Management}
\\[0.2em]
{\small \textbf{Applied Mathematical Finance, Forthcoming}\tnoteref{t1}}
}
\tnotetext[t1]{The authors would like to thank NSERC for partially funding this work.}

\author{Ali Al-Aradi}
\ead{ali.al.aradi@utoronto.ca}

\author{Sebastian Jaimungal}
\ead{sebastian.jaimungal@utoronto.ca}
\address{Department of Statistical Sciences, University of Toronto}

\begin{abstract}
Portfolio management problems are often divided into two types: active and passive, where the  objective is to outperform and track a preselected benchmark, respectively. Here, we formulate and solve a dynamic asset allocation problem that combines these two objectives in a unified framework. We look to maximize the expected growth rate differential between the wealth of the investor's portfolio and that of a \textit{performance} benchmark while penalizing risk-weighted deviations from a given \textit{tracking} portfolio. Using stochastic control techniques, we provide explicit closed-form expressions for the optimal allocation and we show how the optimal strategy can be related to the growth optimal portfolio. The admissible benchmarks encompass the class of functionally generated portfolios (FGPs), which include the market portfolio, as the only requirement is that they depend only on the prevailing asset values. Finally, some numerical experiments are presented to illustrate the risk-reward profile of the optimal allocation. \\
\end{abstract}


\end{frontmatter}

\setlength{\parskip}{0.5cm} 
\setlength{\parindent}{1cm}
\titlespacing{\section}{0pt}{\parskip}{-0.5\parskip}
\titlespacing{\subsection}{0pt}{\parskip}{-0.5\parskip}
\titlespacing{\subsubsection}{0pt}{\parskip}{-0.5\parskip}

\section{Introduction}
\textit{Active portfolio management} aims at constructing portfolios that generate superior returns. This is in contrast to \textit{passive portfolio management}, where the goal is to track a given index. The goals of active management can be further divided into those of an \textit{absolute} nature and those of a \textit{relative} nature. Absolute goals do not involve any external processes, e.g. maximizing expected growth, minimizing the variance of terminal wealth or the probability of ruin. Relative goals are those that involve an external benchmark often given in the form of a portfolio of assets, e.g. maximizing the probability of outperforming the market portfolio. In some cases, the performance of an active manager may also incorporate their deviation from the benchmark, as they are penalized for tracking error or taking on excessive active risk. It may also be desirable to hold portfolios with minimal levels of absolute risk, measured in terms of the volatility of the portfolio's wealth process. The main goal of this paper is to construct portfolios that achieve optimal relative returns against a \textit{performance benchmark} given by a portfolio that depends on prevailing asset values. We also allow for an investor to anchor their portfolio to a \textit{tracking benchmark} by penalizing deviations from the latter. It is often the case for active managers that these two portfolios coincide. Care must be taken in setting up and solving the control problem considering that the weights of the benchmark process are themselves stochastic processes, and must be incorporated as additional state variables.

The impetus for our work is \cite{oderda2015}, where the author investigates optimizing relative returns using stochastic portfolio theory (SPT), but only partially solves the problem via static optimization. Here, we significantly improve on those results by using optimal stochastic control techniques. We succeed in providing explicit closed form expressions for the optimal allocation. In addition, we demonstrate how the optimal portfolio we derive relates to the growth optimal portfolio, which maximizes expected growth over any horizon and plays an important role in the financial literature.

One major issue that is faced by investors that seek outperformance is the need to robustly estimate the growth rates of individual assets. This can be a very difficult goal to achieve. However, allowing the investor to anchor to a portfolio of their choice embeds in our stochastic control problem a portfolio tracking problem which can be used to overcome the issues of estimation. In particular, an investor can choose to track a portfolio known to have certain outperformance properties, such as certain functionally generated portfolios or universal portfolios, and attempt to improve their risk-adjusted performance via the stochastic control problem that we pose.

\section{Literature Review}

There is a great deal of literature dedicated to solving various portfolio selection problems via stochastic control theory. The seminal work of \cite{Merton69} introduced the dynamic asset allocation and consumption problem, utilizing stochastic control techniques to derive optimal investment and consumption policies. Extensions can be found in \cite{Merton71}, \cite{Magill76}, \cite{Davis90}, \cite{Browne97} and more recently \cite{BlanchetScaillet2008}, \cite{Liu2013} and \cite{Ang2014}  to name a few. The focus in these papers is generally on maximizing the utility of discounted consumption and terminal wealth or minimizing shortfall probability, or other related absolute performance measures that are independent of any external benchmark or relative goal.

The question of optimal \textit{active management} was introduced in \cite{Browne1999a} and later refined in \cite{Browne1999b}. In these papers, an investor is   able to trade in a number of assets modeled by geometric Brownian motions (GBMs) and a risk-free asset. The investor aims to maximize the performance of their portfolio \textit{relative} to a stochastic benchmark which is also modeled as an exogenously given GBM. They assume markets to be ``incomplete'' by modeling the benchmark with a Brownian motion that is independent of the risky securities. This modeling approach allows for a wide range of benchmarks, including inflation, exchange rates and other portfolios. The paper provides a general result and is applied to find the optimal portfolio strategy for a number of active management problems, namely: (i) maximizing the probability that the investor's wealth achieves a certain performance goal relative to the benchmark, before falling below it to a predetermined shortfall; (ii) minimizing the expected time to reach the performance goal; (iii) maximizing the expected reward obtained upon reaching the goal;  (iv) minimizing the expected penalty paid upon falling to the shortfall level; (v) maximizing the utility of relative wealth. \cite{Browne2000} extends the work to include a risk constraint in the optimization problem, although the setup there is restricted to a complete market where the benchmark portfolio consists only of assets from the asset universe. \cite{Browne2000} has the same objectives as \cite{Browne1999a} and \cite{Browne1999b}, with the additional constraint that the probability of success is bounded below by a given level.

One important shortcoming in these papers, particularly regarding the problem of beating a benchmark portfolio, is that the solutions only apply to benchmark portfolios with constant weights and can be extended at most to the case of deterministic weight processes. This is because the case of stochastic weights would require including the benchmark weights as state variables, hence significantly increasing the dimensionality of the control problem. It can be argued that the case of stochastic weights is more relevant, considering that one of the simplest and arguably most widespread benchmark is the market portfolio, whose weights evolve stochastically through time.

Much of the model setup in this paper is based on stochastic portfolio theory (see e.g., \cite{fernholz2002} or \cite{karatzas2009} for a more recent overview), which is a flexible framework for analyzing market structure and portfolio behavior. SPT is a descriptive, rather than a normative, approach to addressing these issues and relies on a minimal set of assumptions that are readily satisfied in real equity markets. One of the seminal papers in the development of SPT is \cite{fernholz1982} who introduce the central tools used in SPT, primarily the notion of excess growth, along with a characterization of long-term behavior of stocks and equity markets and the conditions for achieving market equilibrium. \cite{fernholz1999} expands on the work of \cite{fernholz1982} by building on the notion of excess growth and introducing the concept of market diversity and entropy as a measure of diversity, and their role in describing stock market equilibrium. \cite{fernholz1999b} further develops the SPT framework by introducing portfolio generating functions, a tool that can be used to construct dynamic portfolios that depend only on the market weights of the asset universe. \cite{fernholz2001} extends the concept of functionally generated portfolios to functions of ranked market weights, which allows for the construction of portfolios based on company size.

A main focus in SPT is determining conditions under which relative arbitrage opportunities - portfolios that are guaranteed to outperform the market portfolio - exist over various time horizons, and methods for constructing such portfolios. This is discussed in works such as \cite{fernholz2005b}, \cite{banner2008}, \cite{pal2013}, \cite{wong2015optimization}, \cite{pal2016geometry} and \cite{fernholz2016}, among others. It should be noted that, unlike the present work, SPT is typically concerned with outperformance in the a.s. sense. An exception to this is the work of \cite{vervuurt2016}, in which machine learning techniques are utilized to achieve outperformance in expectation by maximizing the investor's Sharpe ratio.

An application of SPT in the context of active management is found in \cite{oderda2015}, where the author seeks to explain the superior risk-adjusted performance of rule-based, non cap-weighted investment strategies relative to the market portfolio. The dynamics of relative wealth are described using results from SPT and a static optimization is applied to solve for the investor's optimal portfolio. The solution to this utility maximization problem yields an optimal portfolio that holds a proportion in five passive rule-based portfolios: the market portfolio, the equal-weight portfolio, the global minimum variance portfolio, the risk parity portfolio and the high-cash flow portfolio.

As mentioned earlier, here, we significantly improve on the results in \cite{oderda2015} by using optimal stochastic control techniques.  We succeed in providing explicit closed form expressions for the optimal allocation. In addition, we demonstrate how the optimal portfolio relates to the growth optimal portfolio and how to apply our techniques to leverage the relative arbitrage properties of certain functionally generated portfolios.

\section{Model Setup} \label{sec:modelSetup}

\subsection{Market Model}

We begin with the typical SPT market model as described in, for example, \cite{fernholz2002}. Let $\vW = \left\{ \vW(t) = (W_1(t),..., W_k(t))' \right\}_{t \geq 0}$ be a standard $k$-dimensional Wiener process defined on the filtered probability space  $(\Omega, \Ft, \mathfrak{F}, \PP)$, where $\mathfrak{F} = \{\Ft_t\}_{t \geq 0}$ is the $\PP$-augmentation of the natural filtration generated by $\vW$, $\Ft_t = \sigma(\{\vW(s)\}_{s \in [0,t]})$. The Wiener processes are the drivers of the assets returns in the economy. Further, we will assume that the market consists of $n$ stocks, where $n \geq k$.
\begin{Definition}
The \textbf{stock price process} for asset $i$, $X_i = \left\{ X_i(t) \right\}_{t \geq 0}$,  is a continuous semimartingale of the form
\begin{equation}
X_i(t) = X_i(0) \cdot \exp \left( \int_0^t \left(\gamma_i(s) + \delta_i(s)\right) ds + \int_0^t \sum_{\nu=1}^{k} \xi_{i\nu}(s) ~ dW_\nu(t) \right)
\end{equation}
for $i = 1, ..., n$, where $\gamma_i$, $\delta_i$ and $\xi_{i\nu}$ are deterministic functions corresponding to the asset's \textbf{growth rate}, \textbf{dividend rate} and \textbf{volatility} with respect to the $\nu$-th source of randomness.
\end{Definition}

We require the following assumption on the model parameters:

\begin{Assumption} \label{asmp:1}
The growth, dividend and volatility functions, $\gamma_i$,  $\delta_i$ and $\xi_{iv}$, are bounded and differentiable for $i = 1,... , n$ and $v = 1,..., k$.
\end{Assumption}

One model which satisfies this assumption is one where the growth, dividend and volatility parameters are constant; a multivariate extension of the typical Black-Scholes model. This model will be used in the implementation section to obtain some numerical results. Note that our market model does not include a risk-free asset, however it would not be difficult to extend the control problem to incorporate such an asset. Additionally, Assumption \ref{asmp:1} is stronger than the usual integrability assumptions made in SPT contexts, but simplifies the process of solving the stochastic control problem via a dynamic programming approach. While it is possible to allow for the model parameters to be stochastic, this will require additional assumptions and it will make the proofs of optimality more cumbersome without adding a great deal to the main results of this paper.

Without loss of generality, we assume each asset has a single outstanding share, so that $X_i$ represents the total market capitalization for the corresponding security. Also, it is more convenient to work with log prices, $\log X_i$, which satisfy the stochastic differential equation (SDE):
\begin{equation} \label{stockSDE}
d \log X_i(t) = \left(\gamma_i(t) + \delta_i(t) \right) dt + \vxi_i(t)' ~ d\vW(t)\,,
\end{equation}
where $\vxi_i(t) = (\xi_{i1}(t), ..., \xi_{ik}(t))'$ is a $k \times 1$ vector of volatilities. This can also be expressed in vector notation as follows:
\begin{equation}
d \log \vX(t) = \left(\vgamma(t) + \vdelta(t) \right) dt + \vxi(t) ~ d\vW(t)\,,
\end{equation}
where
\begin{align*}
\underset{\color{red} (n \times 1)}{\log \vX(t)} &= \left(\log X_1(t), ... , \log X_n(t) \right)'\,,
&& \underset{\color{red} (n \times 1)}{\vgamma(t)} = (\gamma_1(t), ... , \gamma_n(t))'\,, \\
\underset{\color{red} (n \times 1)}{\vdelta(t)} &= (\delta_1(t), ... , \delta_n(t))'\,,
&& \underset{\color{red} (n \times k)}{\vxi(t)} = (\vxi_1(t), ... , \vxi_n(t))'\,.
\end{align*}

Next, we introduce the covariance process:
\begin{Definition}
The \textbf{covariance process} is a matrix-valued function given by:
\begin{equation} \label{eqn:sigma}
\underset{\color{red} (n \times n)}{\mSigma(t)} = \vxi(t) \cdot \vxi(t)'\,.
\end{equation}
\end{Definition}
The covariance of asset $i$ with asset $j$ can then be expressed as $\mSigma_{ij}(t) = \vxi_i(t)'\, \vxi_j(t)$. We assume next that the market satisfies the usual \textbf{nondegeneracy condition} that ensures that the covariance matrix is invertible:
\begin{Assumption}
The covariance process $\mSigma(t)$ is nonsingular for all $t \geq 0$, and there exists $\varepsilon_1 > 0$ such that for all $\mathbf{x} \in \RR^n$ and $t \geq 0$
\begin{equation}
\mathbf{x}' \mSigma(t) \mathbf{x} ~\geq~ \varepsilon_1 \| \mathbf{x} \|^2 \qquad \PP\mbox{-a.s.}
\end{equation}
\end{Assumption}
Thus, $\mSigma(t)$ is positive definite for all $t \geq 0$, and $\mSigma^{-1}(t)$ exists for all $t \geq 0$.
We also assume the market has \textbf{bounded variance}:
\begin{Assumption}
There exists $M_1 > 0$ such that for all $\mathbf{x} \in \RR^n$ and $t \geq 0$,
\begin{equation}
\mathbf{x}' \mSigma(t) \mathbf{x} ~\leq~ M_1 \| \mathbf{x} \|^2 \qquad \PP\mbox{-a.s.}
\end{equation}
\end{Assumption}

\subsection{Portfolio Dynamics}
\begin{Definition}
A \textbf{portfolio} is a measurable, $\mathfrak{F}$-adapted, vector-valued process $\vpi = \{ \vpi(t) \}_{t \geq 0}$, where $\vpi(t) = \left(\pi_1(t), ..., \pi_n(t) \right)$ such that, for all $t \geq 0$, $\vpi(t)$ is bounded $\PP$-almost surely and satisfies:
\begin{equation}
\pi_1(t) + \cdots + \pi_n(t) = 1 \quad \PP\mbox{-a.s.}
\end{equation}
\end{Definition}

Each component of $\vpi$ represents the proportion of wealth invested in the corresponding stock; negative values of $\pi_i(t)$ indicate a short position in stock $i$. From a well-known result in SPT, the logarithm of the \textbf{portfolio value process}, $Z_\vpi = \{ Z_\vpi(t) \}_{t \geq 0}$, satisfies the SDE
\begin{equation} \label{portSDE}
d \log Z_\vpi(t) = \left(\gamma_\vpi(t) + \delta_\vpi(t) \right) ~dt + \vxi_\vpi(t)' ~ d\vW(t)\,,
\end{equation}
\begin{align*}
\text{where} ~~~ \underset{\color{red} (1 \times 1)}{\gamma_\vpi(t)} &= \vpi(t) ' \vgamma(t) + \gamma^*_\vpi(t)\,,
&& \underset{\color{red} (1 \times 1)}{\gamma^*_\vpi(t)}  = \tfrac{1}{2} \left[ \vpi(t)' \mbox{diag}(\mSigma(t)) - \vpi(t)' \mSigma(t) \vpi(t)  \right]\,, \\
\underset{\color{red} (1 \times 1)}{\delta_\vpi(t)} & = \vpi(t)' \vdelta(t)\,,
&& \underset{\color{red} (k \times 1)}{ \vxi_\vpi(t)} =  \vxi(t)' \vpi(t)\,.
\end{align*}

\noindent Here, $\gamma_\vpi$, $\delta_\vpi$ and $\vxi_\vpi$ are the \textbf{portfolio growth rate}, \textbf{portfolio dividend} and \textbf{portfolio volatility} processes, respectively, and $\gamma^*_\vpi$ is the \textbf{excess growth rate} process of the portfolio $\vpi$. The excess growth rate is equal to half of the difference between the portfolio-weighted average of asset return variances and the portfolio variance. This is a pivotal quantity which plays an important role in SPT, as discussed in more detail in \cite{fernholz2002}.

One portfolio of particular interest to us is the \textbf{market portfolio}. The market portfolio plays a central role in SPT and, in our context, is one of the most common benchmarks assigned to active managers.
\begin{Definition}
The \textbf{market portfolio} (process), $\vmu = \{\vmu(t)\}_{t \geq 0}$, is the portfolio with weights given by:
\begin{equation}  \label{def:MW}
\quad \quad \mu_i(t) = \frac{X_i(t)}{X_1(t) + \cdots + X_n(t)} = \frac{X_i(t)}{Z_\vmu(t)}\,, \qquad \mbox{for } i = 1,..., n\,.
\end{equation}
\end{Definition}
In other words, the market portfolio holds each stock according to its proportional market capitalization.
\begin{Remark}
Note that the market portfolio is a vector-valued stochastic process. It is important to highlight the fact that it is a stochastic process and its dynamics must be taken into account when setting up and solving our optimization problem.
\end{Remark}

\subsection{Relative Return Dynamics}
The main state variable in our optimization problem will be the ratio of the wealth of an arbitrary portfolio $\vpi$ relative to that of a preselected performance benchmark $\vrho$. The following proposition specifies the dynamics of the logarithm of this process.

\begin{Proposition}\label{prop:Y_SDE}
Let $Y^{\vpi,\vrho}(t) = \log \left(\frac {Z_\vpi(t)}{Z_\vrho(t)} \right)$ denote the logarithm of relative portfolio wealth for the portfolios $\vpi$ and $\vrho$. Then this process satisfies the SDE:
\begin{align} \label{relRet}
dY^{\vpi,\vrho}(t) &=~ a(t,\vrho(t),\vpi(t)) ~dt + \left( \vxi_\vpi (t) - \vxi_\vrho(t) \right)' ~ d\vW(t)\;,
\end{align}
where
\begin{align*}
a(t,\vrho(t),\vpi(t)) ~&=~ \left(\gamma_\vpi(t) + \delta_\vpi(t) \right) - \left(\gamma_\vrho(t) + \delta_\vrho(t) \right) \\
~&=~ \vpi(t)' \valpha(t) - \tfrac{1}{2} \vpi(t)' \mSigma(t) \vpi(t) - \left(\gamma_\vrho(t) + \delta_\vrho(t) \right) \\
\valpha(t) ~&=~ \vgamma(t) + \vdelta(t) + \tfrac{1}{2} \mbox{diag}\left( \mSigma(t) \right)
\end{align*}
\end{Proposition}
\begin{proof}
The proof follows from noticing that $Y^\pi(t)  = \log Z_\vpi(t) - \log Z_\vrho(t)$, and then applying the portfolio dynamics given in Equation \eqref{portSDE}.
\end{proof} \\

\noindent Note that in the proposition above $\valpha(t) = (\alpha_1(t), ... ,\alpha_n(t))'$ is a vector of the instantaneous \textbf{rate of return} for each asset in the economy.

\begin{Corollary}\label{coro:quadratic_var}
The quadratic variation of $Y^{\vpi,\vrho}$ is given by:
\begin{align}
d\left\langle Y^{\vpi,\vrho}, Y^{\vpi,\vrho} \right\rangle_t &= \left(\vpi(t) - \vrho(t)\right)' \mSigma(t) \left(\vpi(t) - \vrho(t)\right) ~dt\;.
\end{align}
\end{Corollary}
\begin{proof}
Using the volatility of $Y^{\vpi,\vrho}$ given in Proposition \ref{prop:Y_SDE} and the definition of $\mSigma$ and $\vxi_\vpi$ in \eqref{eqn:sigma} and \eqref{portSDE}, respectively, we have:
\begin{align*}
d\left\langle Y^{\vpi,\vrho}, Y^{\vpi,\vrho} \right\rangle_t &= \left( \vxi_\vpi (t) - \vxi_\vrho(t) \right)' \left( \vxi_\vpi (t) - \vxi_\vrho(t) \right) ~ dt \\
&= \left(\vpi(t) - \vrho(t)\right)' \vxi(t) \cdot \vxi(t)' \left(\vpi(t) - \vrho(t)\right) ~dt \\
&= \left(\vpi(t) - \vrho(t)\right)' \mSigma(t) \left(\vpi(t) - \vrho(t)\right) ~dt
\end{align*}
which is the required result. \hfill \end{proof}

In the context of active management this quantity is typically referred to as (instantaneous) \textbf{active risk} or \textbf{tracking error}. An active manager is typically penalized for incurring excessive active risk over their investment horizon.

\begin{Remark}
\cite{fernholz2002} refers to this quantity as the relative variance process of $\vpi$ versus $\vrho$, and denotes by $\tau_{\vpi \vpi}^\vrho(t)$. Note that this quantity is equal to zero if and only if $\vpi$ and $\vrho$ are equal.
\end{Remark}

\section{Stochastic Control Problem} \label{sec:stochasticControl}
In this section we formulate the main stochastic control problem. First, we fix two portfolios against which we measure our outperformance and our active risk, respectively. That is, we choose a \textbf{performance benchmark} $\vrho$, which the investor wishes to outperform, and a \textbf{tracking portfolio} $\veta$, which the investor will penalize deviations from.

\begin{Remark}
The common setup in active portfolio management is for the performance benchmark to be the same as the tracking portfolio. However, this general setup allows for added flexibility and certainly accommodates the case where $\vrho = \veta$. This separation is also useful for the situation where we remove the performance benchmark and track functionally generated portfolios to achieve our outperformance goal - this will be discussed in more detail later in the paper.
\end{Remark}

The primary objective is to determine the portfolio process $\vpi$ that maximizes the expected growth rate of relative wealth to $\vrho$ over the investment horizon $T$. This is equivalent to maximizing the expected utility of relative return assuming a log-utility function. Additionally, the investor is penalized for taking on excessive levels of active risk (measured against $\veta$). We also incorporate a general quadratic penalty term that is independent of the two benchmarks. Putting this together, the \textbf{performance criteria} of a portfolio $\vpi$ is given by:
\small
\begin{align} \label{eqn:perfCrit}
H^\vpi(t,y,\vx) = \EE_{t,y,\vx} \biggl[ \zeta_0 \cdot Y^{\vpi,\vrho}(T) - \frac{\zeta_1}{2} \int_t^T (\vpi(s) - \veta(s))' \mSigma(s) (\vpi(s) - \veta(s)) ~ ds  - \frac{\zeta_2}{2}  \int_t^T \vpi(s)' \mQ(s) \vpi(s) ~ds \biggr]\;,
\end{align}
\normalsize
where $\zeta_0,\zeta_1,\zeta_2 \geq 0$ and $\EE_{t,y,\vx}[ ~\cdot~ ]$ is shorthand for denoting $\EE[ ~\cdot~ | Y^{\vpi,\vrho}(t) = y, \vX(t) = \vx ]$. We also need the following assumptions:

\begin{Assumption}
The matrix-valued function $\mQ$ is symmetric, positive-definite, (element-wise) differentiable and satisfies
\[ \varepsilon_2 \| \mathbf{x} \|^2 ~\leq~ \mathbf{x}' \mQ(t) \mathbf{x} ~\leq~ M_2 \| \mathbf{x} \|^2 \]
for some $\varepsilon_2 > 0$ and $M_2 < \infty$ for all $\mathbf{x} \in \RR^n$ and $t \geq 0$.
\end{Assumption}

\begin{Assumption} \label{asmp:markov}
The performance benchmark process $\vrho = \left\{ \vrho(t) \right\}_{t \geq 0}$ and the tracking portfolio process $\veta = \left\{ \veta(t) \right\}_{t \geq 0}$ are Markovian in $\vX$; i.e. there exist bounded functions $\rho: [0,T] \times \RR^n_+ \to \RR^n$ and $\eta: [0,T] \times \RR^n_+ \to \RR^n$ such that $\vrho(t) = \rho(t,\vX(t))$ and  $\veta(t) = \eta(t,\vX(t))$. Additionally, we assume that the functions $\rho$ and $\eta$ are once differentiable with respect to all of their variables.
\end{Assumption}
Note that the assumption above is satisfied by the class of (time-dependent) functionally generated portfolios, which includes the ubiquitous market portfolio, as they depend only on the prevailing market weights which in turn depend only on the asset values.

Let us now discuss the components of the performance criteria \eqref{eqn:perfCrit}:
\begin{enumerate}
\item The first term, $\zeta_0 Y^{\vpi,\vrho}(T)$, is a terminal reward term which corresponds to the investor wishing to \textbf{maximize the expected growth rate differential between their portfolio and the performance benchmark $\vrho$}. It is also equivalent to \textbf{maximizing the expected utility of relative wealth assuming a log-utility function}. The constant $\zeta_0$ is a subjective parameter determining the amount of emphasis the investor wishes to place on benchmark outperformance. Setting $\zeta_0$ equal to 0 changes our portfolio problem to a tracking problem where the goal is to track the benchmark while minimizing some running quadratic penalty term. This will become particularly useful when we choose to track a functionally generated portfolio to achieve outperformance as it allows us to avoid estimating asset growth rates which can be extremely difficult.
\item The second term is a running penalty term which \textbf{penalizes tracking error/active risk}. That is, deviations from the tracking portfolio are penalized on a running basis, with deviations in riskier assets being penalized more heavily. Alternatively, this penalty term can be viewed as an attempt to \textbf{minimize the quadratic variation of relative wealth with respect to $\veta$, $Y^{\vpi,\veta}(t)$}. The parameter $\zeta_1$ represents the investor's tolerance for active risk.
\item The final term is a \textbf{general quadratic running penalty term} - that does not involve either benchmark - with an associated tolerance parameter $\zeta_2$. One possible choice for $\mQ(t)$ is the covariance matrix $\mSigma(t)$, which can be adopted to minimize the \textit{absolute} risk of the portfolio measured in terms of the quadratic variation of the portfolio wealth process, $Z_\vpi(t)$. Another option is to let $\mQ(t)$ be a constant diagonal matrix, which has the effect of penalizing allocation in each asset according to the magnitude of the corresponding diagonal entry. The investor can use this choice of $\mQ$ as a way of imposing a set of ``soft'' constraints on allocation to each asset. We will demonstrate later on how this penalty term can be interpreted as tilting away from certain assets or, more generally, ``shrinking'' towards some desired portfolio.
\end{enumerate}

\begin{Remark}
While it may be possible to extend this control problem by considering alternative utility functions, particularly to observe the effect of the performance benchmark on the optimal portfolio, this extension does not appear to be trivial. Using general utility functions would likely require other approaches
and lies outside the scope of this work. \\
\end{Remark}

Now, the performance criteria \eqref{eqn:perfCrit} can be simplified by noticing that:
\[ Y^{\vpi,\vrho}(T) = Y^{\vpi,\vrho}(t) + \int_t^T a(s,\vrho(s),\vpi(s)) ~ds  + \int_t^T \left( \vxi_\vpi (t) - \vxi_\vrho(t) \right)' ~ d\vW(s) \]
and that, since $\vxi(t)$ is square-integrable, the conditional expectation of the second integral above is equal to zero due to the martingale property of the stochastic integral. This allows us to rewrite the performance criteria as:

\small
\begin{align} \label{eqn:perfCrit2}
H^\vpi(t,y,\vx) = \zeta_0 y + \EE_{t,\vx} \biggl[ \int_t^T & \zeta_0 \cdot a(s,\vrho(s),\vpi(s) )  \nonumber \\
& - \tfrac{\zeta_1}{2} \cdot (\vpi(s) - \veta(s))' \mSigma(s) (\vpi(s) - \veta(s))  - \tfrac{\zeta_2}{2} \cdot \vpi(s)' \mQ(s) \vpi(s) ~ds \biggr] \nonumber \\
& \hspace{-3.3cm} = \zeta_0 y + h^\vpi(t,\vx)
\end{align}
\normalsize
where $h^\vpi(t,\vx)$ can be interpreted as the performance criteria for the related stochastic control problem consisting only of the running reward/penalty term appearing in the integral of \eqref{eqn:perfCrit2} above. The \textbf{value function} for this auxiliary stochastic control problem is given by:
\begin{equation} \label{control}
h(t,\vx) = \underset{\vpi \in \As}{\sup} ~ h^\vpi(t,\vx)\;,
\end{equation}
where $\As$ is the set of admissible strategies consisting of all portfolios, which, in our context is defined to be the collection of almost surely bounded vector-valued processes that sum up to 1, and hence all portfolios have finite $\mathbb{L}^2(\Omega \times [0,T])$-norm. Notice that, due to the assumptions made in Section \ref{sec:modelSetup}, the performance criteria \eqref{eqn:perfCrit2} is finite for any portfolio in the admissible set. Assuming the appropriate differentiability conditions, the dynamic programming principle suggests that the value function in (\ref{control}) satisfies the Hamilton-Jacobi-Bellman (HJB) equation:
\begin{equation} \label{eqn:HJB}
\begin{cases}
\partial_t h + \underset{\vpi \in \As}{\sup} \left\{ \LL^\vX h + \zeta_0 \cdot a(t,\rho(t,\vx),\vpi) - \tfrac{\zeta_1}{2} \cdot (\vpi - \eta(t,\vx))' \mSigma(t) (\vpi - \eta(t,\vx)) - \tfrac{\zeta_2}{2} \cdot \vpi' \mQ(t) \vpi \right\} = 0\;,\\
\hspace{14.92cm} h(T,\vx) = 0\;.
\end{cases}
\end{equation}
where $\LL^\vX$ is the infinitesimal generator of the asset value processes. Notice that Assumption \ref{asmp:markov} ensures that the asset values are the only state variables in this PDE. \\

\begin{Proposition} \label{prop:valueFn}
The solution to the HJB equation \eqref{eqn:HJB} is given by:
\begin{equation} \label{valueFn}
h(t,\vx) ~=~ \EE_{t,\vx} \left[ -\int_t^T G(s,\vX(s)) ~ds  \right]  \;,
\end{equation}
where the expectation is taken under the physical measure $\PP$, and the function $G$ is given by:
\begin{align*} \small
G(t,\vx) =  C(t,\vx) + \frac{1}{2} \left[ \frac{1 - \ones' \mA^{-1}(t) \vB(t,\vx)}{\ones' \mA^{-1}(t) \ones} \cdot \ones + \vB(t,\vx)  \right]' \mA^{-1}(t) \left[  \frac{1 - \ones' \mA^{-1}(t) \vB(t,\vx)}{\ones' \mA^{-1}(t) \ones} \cdot \ones - \vB(t,\vx) \right] \;.
\end{align*}
with
\begin{align*}
\underset{\color{red} (n \times n)}{\mA(t)} &=  (\zeta_0 + \zeta_1) \mSigma(t) + \zeta_2 \mQ(t) \;,
 \\
\underset{\color{red} (n \times 1)}{\vB(t,\vx)} &= \zeta_0 \valpha(t) + \zeta_1 \mSigma(t) \eta(t,\vx) \;,
\quad \text{and}
\\
\underset{\color{red} (1 \times 1)}{C(t,\vx)} &= \zeta_0 \cdot \left(\vgamma_{\rho(t,\vx)}(t) + \vdelta_{\rho(t,\vx)}(t) \right) + \tfrac{\zeta_1}{2} \cdot \eta(t,\vx)' \mSigma(t) \eta(t,\vx) \;.
\end{align*}
\end{Proposition}
\begin{proof}
See Appendix \ref{proof:valueFn}.
\end{proof}

Next, we present one of the key results of the paper: a verification theorem showing that the candidate solution provided in Proposition \ref{prop:valueFn} coincides with the value function. As well, we present the explicit form of the optimal weights.

\begin{Theorem} \label{thm:optCont}
The candidate value function given in Proposition \ref{prop:valueFn} is indeed the solution to the stochastic control problem \eqref{control}. Moreover, the optimal portfolio is given by:
\begin{subequations}
\label{eqn:opt-pi-A-and-B}
\begin{equation} \label{eqn:optCont}
\vpi_\vzeta^*(t) = \mA^{-1}(t) \cdot \left[ \frac{1 - \ones' \mA^{-1}(t) \cdot \vB(t,\vX(t))}{\ones' \mA^{-1}(t) \ones} \cdot \ones + \vB(t,\vX(t)) \right]\;,
\end{equation}
where
{\small
\begin{equation}
\underset{\color{red} (n \times n)}{\mA(t)} =  (\zeta_0 + \zeta_1) \mSigma(t) + \zeta_2 \mQ(t)\;, \quad
\text{and} \quad
\underset{\color{red} (n \times 1)}{\vB(t,\vx)} = \zeta_0 \valpha(t) + \zeta_1 \mSigma(t) \eta(t,\vx) \;.
\end{equation} \vspace{0.05cm}
}
\end{subequations}
\end{Theorem}
\begin{proof}
See Appendix \ref{proof:optCont}.
\end{proof}

If an investor is faced with the market portfolio as their performance and tracking benchmark, as is the case with the vast majority of active management mandates, then the optimal strategy above can be applied by replacing $\eta$ with the function $\mu:\RR^n_+ \rightarrow \Delta^n$:
\begin{equation} \label{eqn:muFunction}
\mu(\vX) = \left( \frac{X_1}{\sum_{i=1}^n X_i}, ..., \frac{X_n}{\sum_{i=1}^n X_i} \right) \;,
\end{equation}
where $\Delta^n = \left\{ \vmu \in \RR^n : \mu_i \in (0,1), i=1,...,n ; ~\vmu'\ones = 1 \right\}$ is the standard $n$-dimensional simplex.

Before discussing the properties of the optimal portfolio in \eqref{eqn:opt-pi-A-and-B}, we first present two portfolios that are of particular importance: the growth optimal portfolio (GOP) and the minimum quadratic variation portfolio (MQP). As the names suggest, the GOP is the portfolio with maximal expected growth over any time horizon while the MQP is the portfolio with the smallest quadratic variation of all portfolios over the investment horizon. Alternatively, the GOP corresponds to the optimal portfolio for an investor that wishes to maximize the expected utility of terminal wealth when their utility function is logarithmic. The GOP plays an important role in financial theory for a variety of reasons; see e.g., \cite{christensen2012} for a comprehensive literature review. In certain limiting cases, the solution to our stochastic control problem can be written in terms of the GOP and MQP. Formally, the GOP and MQP are the solutions to the following optimization problems:

\vspace{-0.6cm}

\begin{align} \label{GOPproblem}
\vpi_{GOP} &= \underset{\pi \in \As}{\arg \sup} ~ \EE \left[ \log \left( \frac{Z_\vpi(T)}{Z_\vpi(0)} \right) \right] \;, \\
\vpi_{MQP} &= \underset{\pi \in \As}{\arg \inf} ~~ \EE \left[ \int_0^T \vpi' \mSigma(t) \vpi ~dt \right] \;.
\end{align}

\vspace{0.5cm}

\begin{Corollary} \label{cor:GOPMQP}
The growth optimal portfolio is given by:
\begin{equation} \label{eqn:GOP}
\vpi_{GOP}(t) = \left( 1 - \ones' \mSigma^{-1}(t) \valpha(t) \right) \cdot \vpi_{MQP}(t) + \mSigma^{-1}(t) \valpha(t) \;.
\end{equation}
where $\vpi_{MQP}$ is the minimum quadratic variation portfolio given by:
\begin{equation} \label{eqn:MQP}
\vpi_{MQP}(t) = \frac{1}{\ones' \mSigma^{-1}(t) \ones} \cdot \mSigma^{-1}(t) \ones \;.
\end{equation}
\end{Corollary}
\begin{proof}
Use the optimal control given in Theorem \ref{thm:optCont} with $\zeta_0 = \zeta_1 = 0$ and $\mQ = \mSigma$ to obtain the MQP, and with $\zeta_1 = \zeta_2 = 0$ to obtain the GOP. \hfill
\end{proof} \\

\begin{Remark}
\cite{oderda2015} refers to the MQP as the global minimum variance (GMV) portfolio. However, this is a misnomer as the latter term is typically reserved for the portfolio that minimizes the variance of terminal wealth, $\VV[Z_\pi(T)]$, for a given investment horizon. Though it should be noted that in a discrete-time, single-period model where the covariance matrix of returns is given by $\mSigma$, the GMV is in fact equal to $\frac{1}{\ones' \mSigma^{-1} \ones} \cdot \mSigma^{-1}\ones$. \\
\end{Remark}

Armed with the definitions of the GOP and MQP, we present the following corollary which summarizes some of the properties of the optimal strategy given in \eqref{eqn:opt-pi-A-and-B}, as well as its relation to the aforementioned portfolios. \\

\begin{Corollary} \label{cor:properties}
The optimal portfolio given in Theorem \ref{thm:optCont} satisfies:
\begin{enumerate}[label=(\roman*)]
\item The limiting portfolios as the preference parameters tend to infinity are given by:
\begin{align*}
\underset{\zeta_0 \to \infty}{\lim} \vpi_\vzeta^*(t) &= \vpi_{GOP}(t) \\
&= \vpi_\vzeta^*(t) \text{ with } \vzeta = (\zeta_0,0,0) \qquad  \text{(no running penalties)} \\
\underset{\zeta_1 \to \infty}{\lim} \vpi_\vzeta^*(t) & = \veta(t) \\
&= \vpi_\vzeta^*(t) \text{ with } \vzeta = (0,\zeta_1,0) \qquad  \text{(tracking penalty only)} \\
\underset{\zeta_2 \to \infty}{\lim} \vpi_\vzeta^*(t) &= \frac{1}{\ones' \mQ^{-1}(t) \ones} \cdot \mQ^{-1}(t) \ones \\
&= \vpi_\vzeta^*(t)  \text{ with } \vzeta = (0,0,\zeta_2)  \qquad \text{(absolute penalty only)}
\end{align*}
\item When $\zeta_0, \zeta_1 > 0$ and $\zeta_2 = 0$ (no absolute penalty):
\[ \vpi_\vzeta^*(t) = c \cdot \vpi_{GOP}(t) + (1-c) \cdot \veta(t) \]
where $c = \frac{\zeta_0}{\zeta_0 + \zeta_1} \in (0,1)$ is a fixed constant.
\item When $\mQ(t) = \mSigma(t)$ (minimize relative and absolute risk):
\[ \vpi_\vzeta^*(t) = c_1 \cdot \vpi_{GOP}(t) + c_2 \cdot \veta(t) + c_3 \cdot  \vpi_{MQP}(t)  \]
where $c_i = \frac{\zeta_i}{\zeta_0 + \zeta_1 + \zeta_2} \in (0,1)$ for $i = 1,2,3$ are fixed constants summing to 1.
\end{enumerate}
\end{Corollary}
\begin{proof}
See Appendix \ref{proof:properties}.
\end{proof}

A few points are worth mentioning at this juncture:
\begin{enumerate}	
	\item The limiting portfolios in the corollary above show that as the emphasis on outperformance increases, the optimal portfolio tends to the GOP. This is a consequence of the GOP having the highest expected growth rate, and hence delivering the highest level of (expected) outperformance. Similarly, as the aversion to tracking error increases, the optimal portfolio tends to the tracking benchmark. This reflects the fact that the investor is heavily penalizing any deviation from the tracking benchmark and decides to hold the latter to avoid any running penalties. Finally, when the third aversion parameter tends to infinity, the investor holds the portfolio that minimizes the absolute running penalty at each point in time. These three limiting cases are equivalent to investors that are \textit{only} interested in outperformance, tracking or reducing an absolute running penalty, which lead to the expected result of holding the GOP, the benchmark itself or the absolute penalty minimizer, respectively.

    \item In the absence of the absolute running penalty, the optimal strategy consists of holding a constant proportion in the tracking benchmark and the remainder in the GOP. The idea is to benefit from the expected growth rate of the GOP while modulating its high levels of active risk by investing in the tracking benchmark. The proportions are determined by the relative importance the investor places on outperformance versus tracking. The fact that a \textit{constant} proportion of wealth is invested in each portfolio is related to the fact that log-utility falls under the class of constant relative risk aversion (CRRA) utility functions, wherein the fraction of wealth invested in risky assets is independent of the level of wealth. Since we are concerned with \textit{active} risk relative to the tracking benchmark in this case, the tracking benchmark plays the role of the risk-free asset.
	
	\item When the investor wishes to minimize both relative and absolute risk, their optimal strategy is to hold constant proportions in the GOP, MQP and tracking benchmark. These holdings contribute to increasing growth, decreasing absolute risk and decreasing active risk, respectively.

    \item The optimal solution does not depend on the performance benchmark $\vrho$. The reason for this is simple: the investor wishes to maximize the difference between the expected growth rates of their portfolio and the performance benchmark. Regardless of the choice of $\vrho$, the way this difference can be maximized is by choosing the portfolio with the highest possible growth rate, i.e. the GOP. More fundamentally, $\vrho$ does not appear since we have that $\log\left(\frac{Z_\vpi(T)}{Z_\vrho(T)}\right) = \log(Z_\vpi(T))-\log(Z_\vrho(T))$ and the second term does not involve the control and hence can be left out of the performance criteria \eqref{eqn:perfCrit2} altogether. This is unlikely to be the case if alternative utility functions are adopted.

     \item The optimal portfolio is myopic, i.e. independent of the investment horizon $T$. This is expected as the GOP can be shown to maximize expected growth at \textit{any} time horizon. This property is also a typical of logarithmic utility functions.

   	\item When the investor is not interested in outperformance, i.e. $\zeta_0 = 0$, the optimal strategy does not require knowledge/estimation of the assets' growth rates $\vgamma$, which can be quite difficult to estimate robustly. Estimation of the covariance matrix, $\mSigma$, is still needed. $\zeta_0$ can also be used to reflect the investor's confidence in their growth rate estimates: a small value of $\zeta_0$ (relative to $\zeta_1$ and $\zeta_2$) is akin to a small allocation in the growth optimal portfolio (relative to the other subportfolios). Furthermore, depending on the nature of the performance benchmark, outperformance can be achieved in an indirect manner. When that benchmark is the market portfolio, as is often the case for active managers, the tools of SPT provide a wide range of functionally generated portfolios that constitute relative arbitrages with respect to the market portfolio and do not require the estimation of asset growth rates; see \cite{fernholz2002} for more details. The investor can simply track one of these portfolios setting $\zeta_0 = 0$ while imposing their own absolute penalty term via the matrix $\mQ(t)$.
   	
   	\item The absolute penalty term forces the optimal strategy towards the ``shrinkage'' portfolio given by $\frac{1}{\ones' \mQ^{-1}(t) \ones} \cdot \mQ^{-1}(t) \ones$. When $\mQ$ is a diagonal matrix $\mQ(t) = \text{diag}\left(w_1(t),...,w_n(t) \right)$, then $\mQ^{-1}(t) = \text{diag}\left(\tfrac{1}{w_1(t)},...,\tfrac{1}{w_n(t)} \right)$. From this, we see that the absolute penalty term forces us to shrink to a portfolio proportional to $\left(\tfrac{1}{w_1(t)},...,\tfrac{1}{w_n(t)} \right)$; when $w_i$ is large $1/w_i$ is small and the shrinkage portfolio allocates less capital to asset $i$. This can be used to tilt away from undesired assets; taking $w_i \rightarrow \infty$ forces the allocation in asset $i$ to zero. Additionally, taking $\mQ = \mI$ penalizes large positions in any asset by shrinking to the equal-weight portfolio, while setting $w_i(t) = \mSigma_{ii}(t)$ forces shrinkage to the risk parity portfolio. 	
   	
\end{enumerate}

The points made above highlight the motivation for including the two running penalties. In principle, when an investor's goal is simply to outperform a performance benchmark, they would hold the GOP to maximize their expected growth. However, it is well-documented that the GOP is associated with very large levels of risk (in terms of portfolio variance) as well as potentially large short positions in a number of assets. Adding the relative and absolute penalty terms mitigates some of this risk. It is also worth noting that the decompositions in Corollary 3 (ii) and (iii) can help guide the choice of the subjective parameters $\zeta_i$ which can be used to express the proportion of wealth the investor wishes to place in the GOP, tracking benchmark and MQP. The simulation results provided later can also give a rough idea of the relative and absolute risk-return characteristics of optimal portfolios for different choices of $\vzeta$ to further refine the investor's decision.

Finally, we can interpret the optimal portfolio as being the growth optimal portfolio for an alternative market:

\begin{Corollary} \label{cor:modifiedSigma}
The optimal portfolio $\vpi^*_\vzeta$ is the growth optimal portfolio for a market with a modified rate of return process $\valpha^*$ and a modified covariance process $\mSigma^*$ given by:
\begin{subequations}
\begin{align}
\valpha^*(t) &= \zeta_0 \valpha(t) + \zeta_1 \mSigma(t) \eta(t,\vX(t))
\\
\mSigma^*(t) &= (\zeta_0 + \zeta_1) \mSigma(t) + \zeta_2 \mQ(t)
\end{align}
\end{subequations}
\end{Corollary}
\begin{proof}
The result follows from direct comparison of the optimal portfolio in Theorem \ref{thm:optCont} with the form of the growth optimal portfolio in Corollary \ref{cor:GOPMQP}.
\end{proof}

The last corollary has an interesting interpretation. Consider first the adjustment term added to the rate of return process, $\zeta_1 \mSigma(t) \eta(t,\vX(t))$. It is straightforward to show that the elements of this vector constitute the quadratic covariation between between each asset and the wealth process associated with the tracking benchmark $\veta$, namely $\langle \log X_i, \log Z_\veta \rangle$; see Section 1.2 of \cite{fernholz2002}. This implies that the investor is modifying the assets' rates of return to reward those assets that are more closely correlated with the portfolio they are trying to track, and these artificial rewards are proportional to the investor's desire to track $\veta$ represented by $\zeta_1$. Moreover, if we consider the case where $\mQ(t)$ is a diagonal matrix, the modification to the covariance matrix amounts to increasing the variance of each asset according to the corresponding diagonal entry of $\mQ(t)$. This in turn makes certain assets less desirable and is tied to the notion of tilting away from those assets as discussed earlier in this section.

\section{Implementation}

In this section we use Fama-French industry data\footnote{~ Source: \url{http://mba.tuck.dartmouth.edu/pages/faculty/ken.french/data_library.html} } to simulate the optimal portfolio (OP) given by \eqref{eqn:opt-pi-A-and-B} for different values of $\vzeta$. We use the simulations to gain an understanding of the portfolios' features and risk-return profiles compared to the other portfolios, namely: the growth optimal portfolio (GOP) given by \eqref{eqn:GOP} and the minimum quadratic variance portfolio (MQP) given by \eqref{eqn:MQP}, as well as the optimal solution derived in \cite{oderda2015} which we will refer to as the maximal drift portfolio (MDP).

\subsection{Data}

The data consists of monthly returns (with and without dividends) for five industries (Consumer, Manufacturing, Hi-Tech, Health, and Other) for the period January 2005 to July 2017. Industry portfolios are constructed by assigning each stock to one of the five groupings listed above according to their SIC code at the time. We then treat the resulting industries as the constituents of our market, i.e. our market consists of five assets that are the industries themselves. The data also contains the number of firms and the average market capitalization of firms for each industry, which allows us to compute a time series of market weights for the five industries.

We work with industry returns rather than individual securities to avoid difficulties that arise from varying investment sets caused by individual companies entering and exiting the market and we restrict ourselves to a market of five industries for computational simplicity and ease of displaying and interpreting results. In our experiments, we found that simulation results are qualitatively similar for various granularity levels, though it should be noted that parameter estimation, particularly of the growth rates, becomes less reliable and robust as the number of assets increases.

\subsection{Parameter Estimation}

For simplicity, we assume that the parameters $\vgamma, \vdelta$ and $\vxi$ are constants and we use the entirety of the dataset for the estimation. Other estimation methods that allow for time-varying parameters are likely more appropriate, however this is beyond the scope of the current work. The parameter estimation procedure involves estimating:
\begin{enumerate}
\item the covariance matrix, $\mSigma$, estimated by the sample covariance of changes in asset log-values
\item the matrix of volatilities, $\vxi$, which is given as the Cholesky decomposition of $\mSigma$
\item the growth rate, $\vgamma$, estimated as the sample mean of changes in asset log-values (ex-dividends)
\item the dividend growth rate, $\vdelta$, estimated as the sample mean of changes in asset log-values (with dividends) minus $\vgamma$.
\end{enumerate}

A summary of the estimated parameters is given in Table \ref{tab:paramsTable} below.

\begin{table}[h] \footnotesize \centering
	\begin{tabular}{ l c c c }
		\toprule[1.5pt]
		\bf{Industry}    				& \bf{Growth}   & \bf{Dividends} &  \bf{Std. Dev.}   \\
		\midrule	\bf{Cnsmr}  		& 7.4 \%  & 2.1 \%  & 12 \%  \\					
					\bf{Manuf}      	& 5.9 \%  & 2.3 \%  & 16 \%  \\
					\bf{HiTec}  		& 8.6 \%  & 1.5 \%  & 16 \%  \\
					\bf{Hlth}  			& 8.3 \%  & 1.9 \%  & 13 \%  \\
					\bf{Other}  		& 3.2 \%  & 1.8 \%  & 19 \%  \\
		\bottomrule[1.5pt]					
	\end{tabular}  ~~~~~
	\begin{tabular}{ l c c c c c }
		\toprule[1.5pt]
					   &  \bf{Cnsmr}   &   \bf{Manuf}  &  \bf{HiTec}  &  \bf{Hlth}  &  \bf{Other}   \\
		\midrule	\bf{Cnsmr}  & 1      &    		    	 				    \\
					\bf{Manuf}  & 0.76   & 1   			 				   		\\
					\bf{HiTec}  & 0.87   & 0.81      & 1   				    	\\
					\bf{Hlth}   & 0.73   & 0.61      & 0.68   & 1       & 	  	\\
					\bf{Other}  & 0.87   & 0.76      & 0.82   & 0.69    & 1     \\
		\bottomrule[1.5pt]					
	\end{tabular}
	\caption{Summary of estimated model parameters: growth rate $\vgamma$, dividend rate $\vdelta$ and standard deviations $\sqrt{\mSigma_{ii}}$ (left); correlation matrix (right). }
	\label{tab:paramsTable}
\end{table}

\subsection{Simulation}

Next, we use the estimated parameters to simulate industry value paths and compute the strategy's performance in the simulated market model. For the simulation, we consider an investment horizon of $T = 5$ years with daily rebalancing time-steps ($\Delta t = 1/252$; 252 days per year so that the total number of steps is $N = \Delta t \cdot T = 1260$ days) and 1,000 simulations. The asset values are initialized using the market weights at the end of the dataset (market weights on July 2017). Furthermore, we will take the market portfolio to be the performance benchmark and the tracking benchmark, i.e. $\vmu = \vrho = \veta$, as is commonly the case in active portfolio management. Also, we take the absolute penalty matrix to be the covariance matrix, i.e. $\mQ = \mSigma$.

For each simulation, we compute weight vectors for the optimal portfolio at each time step using the estimated model parameters and using various $\vzeta$ values, namely: $\zeta_0 \in \{0.1, 0.5, 5\}$ and $\zeta_1, \zeta_2 \in \{0.001,0.05,0.1,	..., 1\}$. We also compute the GOP and MQP; note that since the model parameters are assumed to be constant, the weight vectors for the GOP and MQP are constant.

\begin{Remark}
The procedure outlined above implicitly assumes that the investor knows the true parameters underlying the data generating process. Clearly, this is not the case in practice and the procedure is more akin to an in-sample form of backtesting. Moreover, in our experiments we find that model misspecification can lead to adverse results as one would expect and it is well-known that parameter estimation, particularly of the asset growth rates, is difficult to perform robustly. However, our point of interest is in studying the features of the optimal portfolio with varying $\vzeta$ values rather than demonstrating its efficacy in actual trading situations. In other words, we wish to separate the parameter estimation problem (which we do not concern ourselves with in this paper) from the stochastic control problem. Additionally, one can heuristically adjust the $\zeta_0$ parameter to reflect the level of ``confidence'' they place in their growth rate estimates.
\end{Remark}

In terms of performance metrics, we will consider the following quantities for each portfolio:
\begin{enumerate}
\item The performance criteria given in \eqref{eqn:perfCrit}, where the running penalties given by the Riemann integrals are approximated using the discretization:
{\small
\begin{align*}
\text{Terminal Reward} ~&=~ \zeta_0 \log\left( \frac{Z_\vpi(T)}{Z_\vmu(T)} \right)\\
\text{Relative Running Penalty} ~&=~ \zeta_1 \Delta t \cdot \sum_{t=1}^N (\vpi(t) - \vmu(t))' ~\mSigma(t)~ (\vpi(t) - \vmu(t))  \\
\text{Absolute Running Penalty} ~&=~ \zeta_2 \Delta t \cdot \sum_{t=1}^N \vpi(t)' ~\mQ(t)~ \vpi(t)
\end{align*}
}
\item The (average) absolute return and active return of each portfolio, defined as the time-series sample mean of portfolio returns and excess returns over the market (for a single path), respectively:
{\small
\begin{align*}
\mbox{Absolute Return} ~&=~ \frac{1}{\Delta t} \frac{1}{N} \sum_{t=1}^N \vpi(t)' ~\boldsymbol{R}(t,t+1) \\
\mbox{Active Return} ~&=~ \frac{1}{\Delta t} \frac{1}{N} \sum_{t=1}^N (\vpi(t) - \vmu(t))' ~\boldsymbol{R}(t,t+1)
\end{align*}}
where $\boldsymbol{R}(t,t+1)$ is the vector of asset returns between $t$ and $t+1$. Note that the factor $\frac{1}{\Delta t}$ is an annualization factor.
\item The absolute and active risk of each strategy, defined as the time-series sample standard deviation of portfolio returns and excess returns over the market, respectively:
{\small
\begin{align*}
\mbox{Absolute Risk} ~&=~ \sqrt {\frac{1}{\Delta t} \frac{1}{N} \sum_{t=1}^N \left( \vpi(t)' ~\boldsymbol{R}(t,t+1) - \mbox{Absolute Return}\right)^2} \\
\mbox{Active Risk} ~&=~ \sqrt {\frac{1}{\Delta t} \frac{1}{N} \sum_{t=1}^N \left( (\vpi(t) - \vmu(t))' ~\boldsymbol{R}(t,t+1) - \mbox{Active Return}\right)^2}
\end{align*} }
\item The Sharpe ratio and information ratio defined as the ratio of absolute return to absolute risk and active return to active risk, respectively:
{\small \begin{align*}
\mbox{Sharpe Ratio} ~&=~ \frac{\mbox{Absolute Return}}{\mbox{Absolute Risk}} \\
\mbox{Information Ratio} ~&=~ \frac{\mbox{Active Return}}{\mbox{Active Risk}}
\end{align*} }
\end{enumerate}

\subsection{Results}

Figure \ref{fig:returnAndRiskSurfaces} shows the average absolute and active risk and return values across all simulations for the optimal strategy with different $\vzeta$ values. We find that as $\zeta_1$ increases, the performance metrics converge to the metrics of the market portfolio; as $\zeta_2$ increases, they converge to the metrics of the MQP (recall that we set $\mQ = \mSigma$). The rate of change is higher for lower values of $\zeta_0$; the optimal strategy is less sensitive to changes in $\zeta_1$ and $\zeta_2$ for high values of $\zeta_0$, in which case it is more closely tied to the GOP.

The convergence features discussed above can be seem more apparently in the scatterplots shown in Figure \ref{fig:riskReturnScatter}, which show the absolute and relative risk-return characteristics of each of the portfolios we analyze. It is apparent that the GOP outperforms all other portfolios in terms of expected return, however, this comes at a high level of absolute risk. On the other hand, we find that the optimal strategies with $\zeta_0 = 0.1$ (shown in red) have a higher return per unit of risk. Additionally, the optimal portfolios dominate the MDP and markt portfolio and trace out an ``efficient frontier'' tying the GOP and MQP. The alignment of the optimal portfolios in the second plot in Figure \ref{fig:riskReturnScatter} is a manifestation of the observation that the information ratio is insensitive to $\zeta_1$. Also, the dots converge to the origin, since the optimal solutions converge to the market portfolio as $\zeta_1$ tends to infinity.

\begin{figure}[h!]
\centering
	\includegraphics[width=0.48\textwidth]{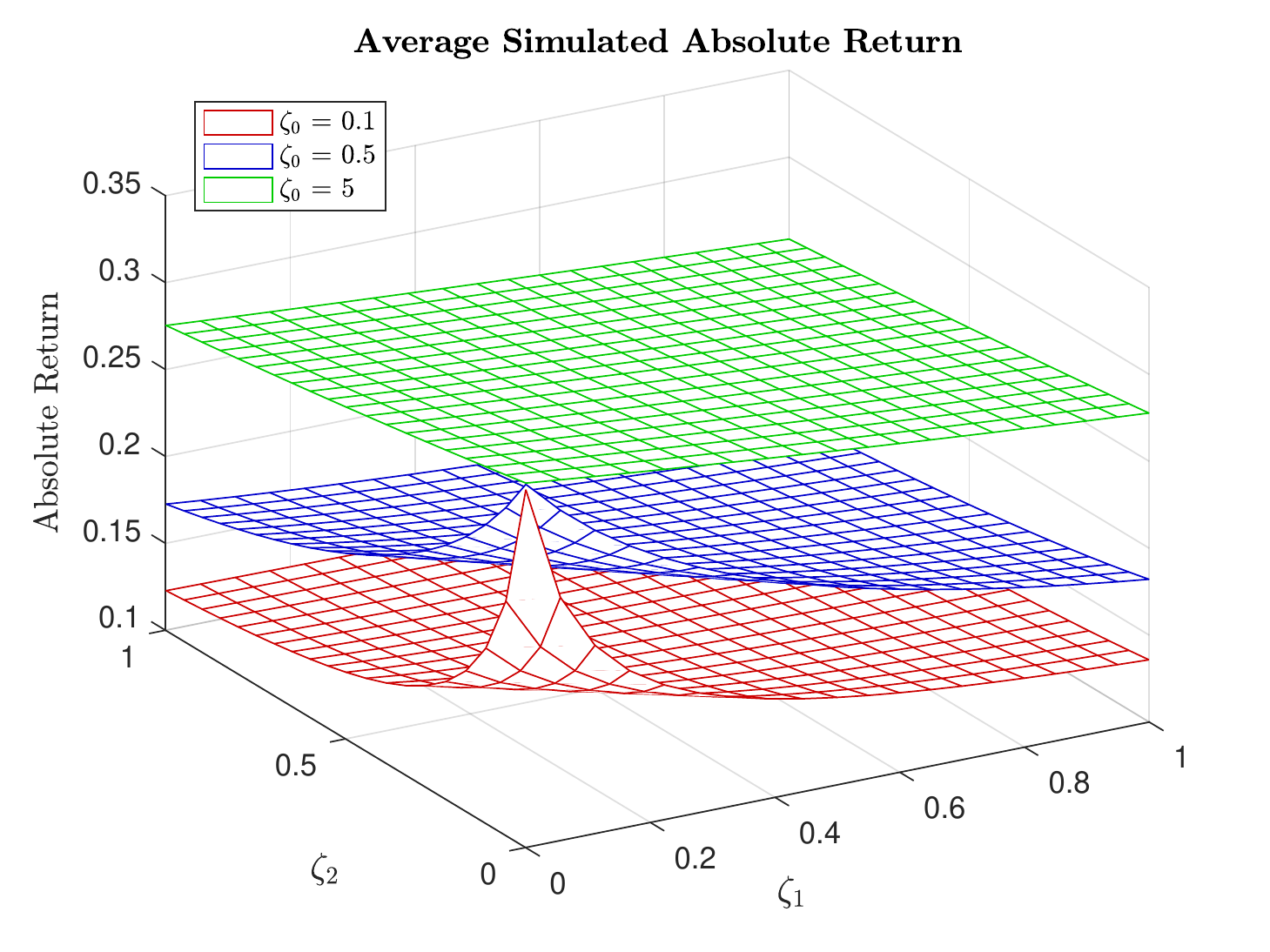}
	\includegraphics[width=0.48\textwidth]{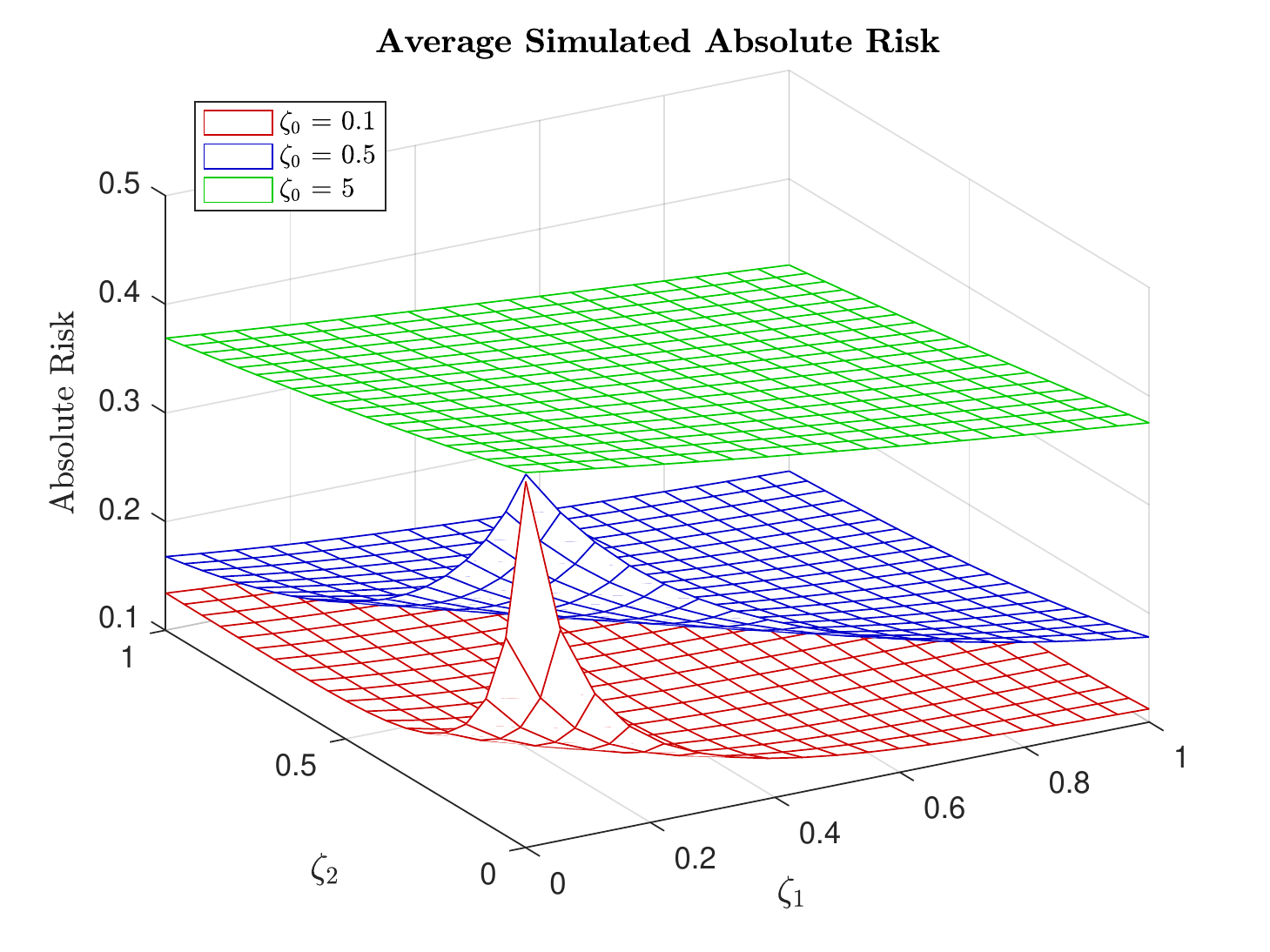}
	\includegraphics[width=0.48\textwidth]{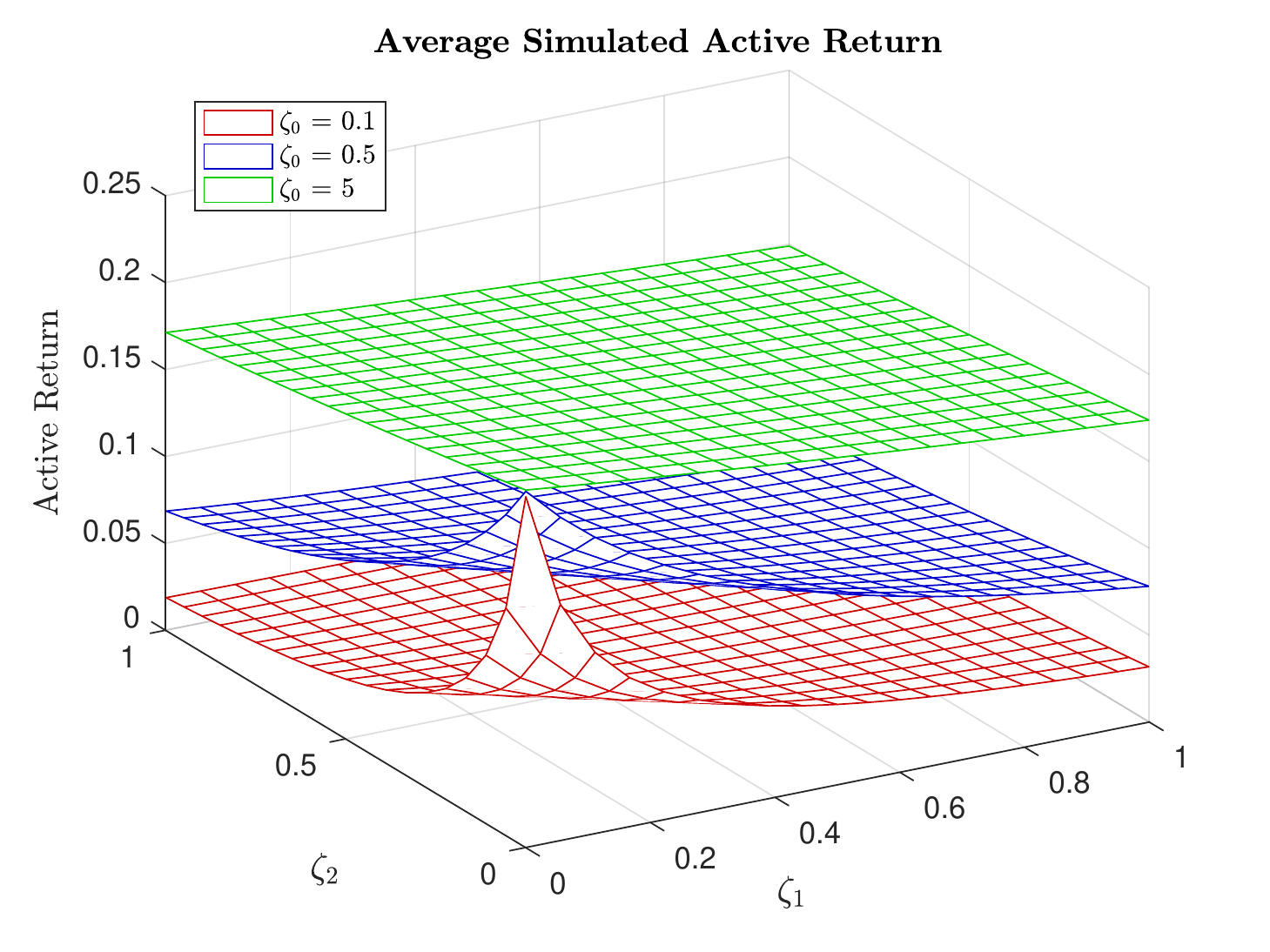}
	\includegraphics[width=0.48\textwidth]{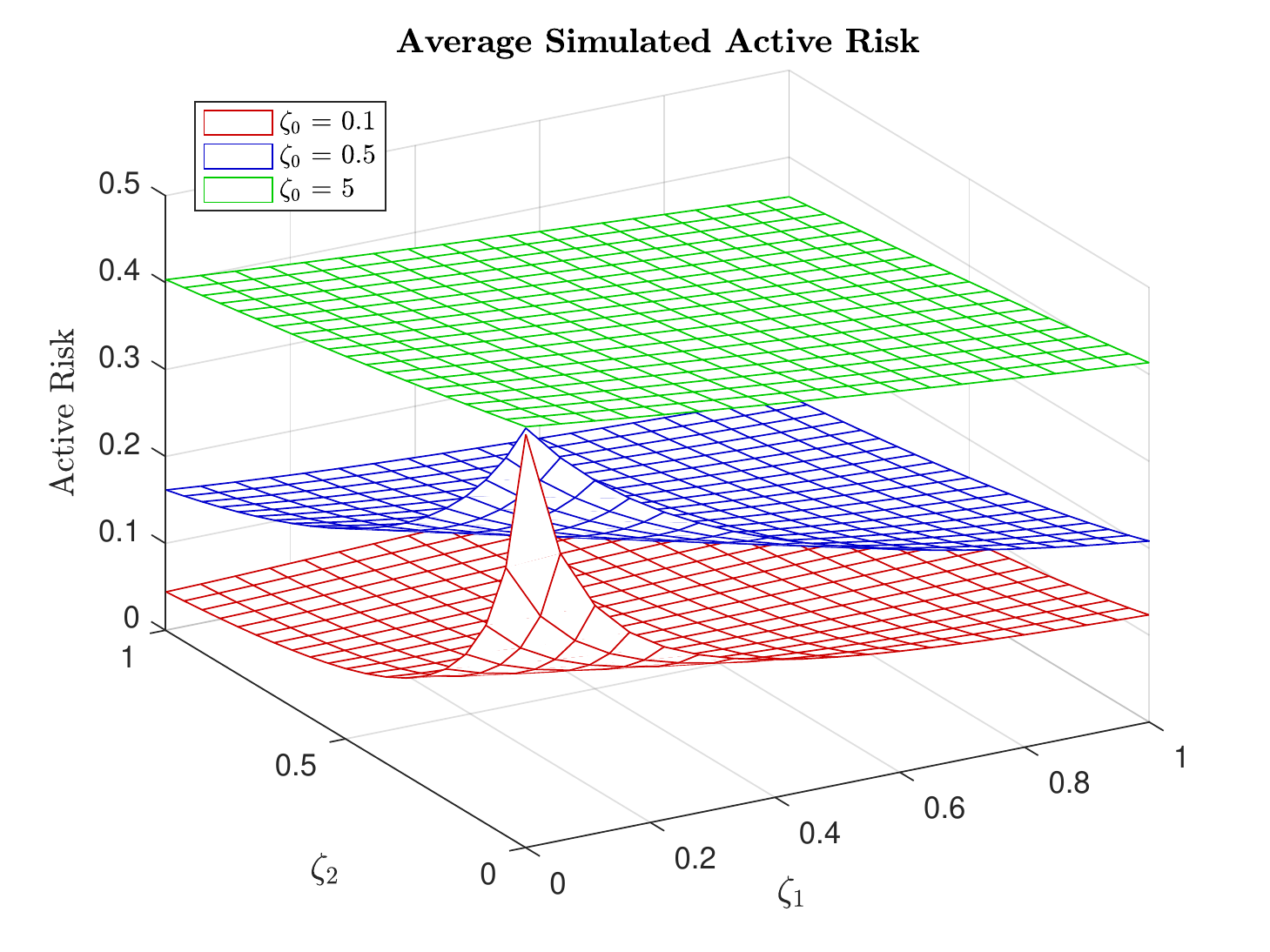}
	\caption{Surfaces of average simulated absolute return and risk (top panels) and active return and risk (bottom panels) for all $\vzeta$ values; each point is the average metric value across all simulations for the optimal portfolio with the corresponding $\vzeta$.}
	\label{fig:returnAndRiskSurfaces}
\end{figure}

\begin{figure}[h!]
\centering
	\includegraphics[width=0.48\textwidth]{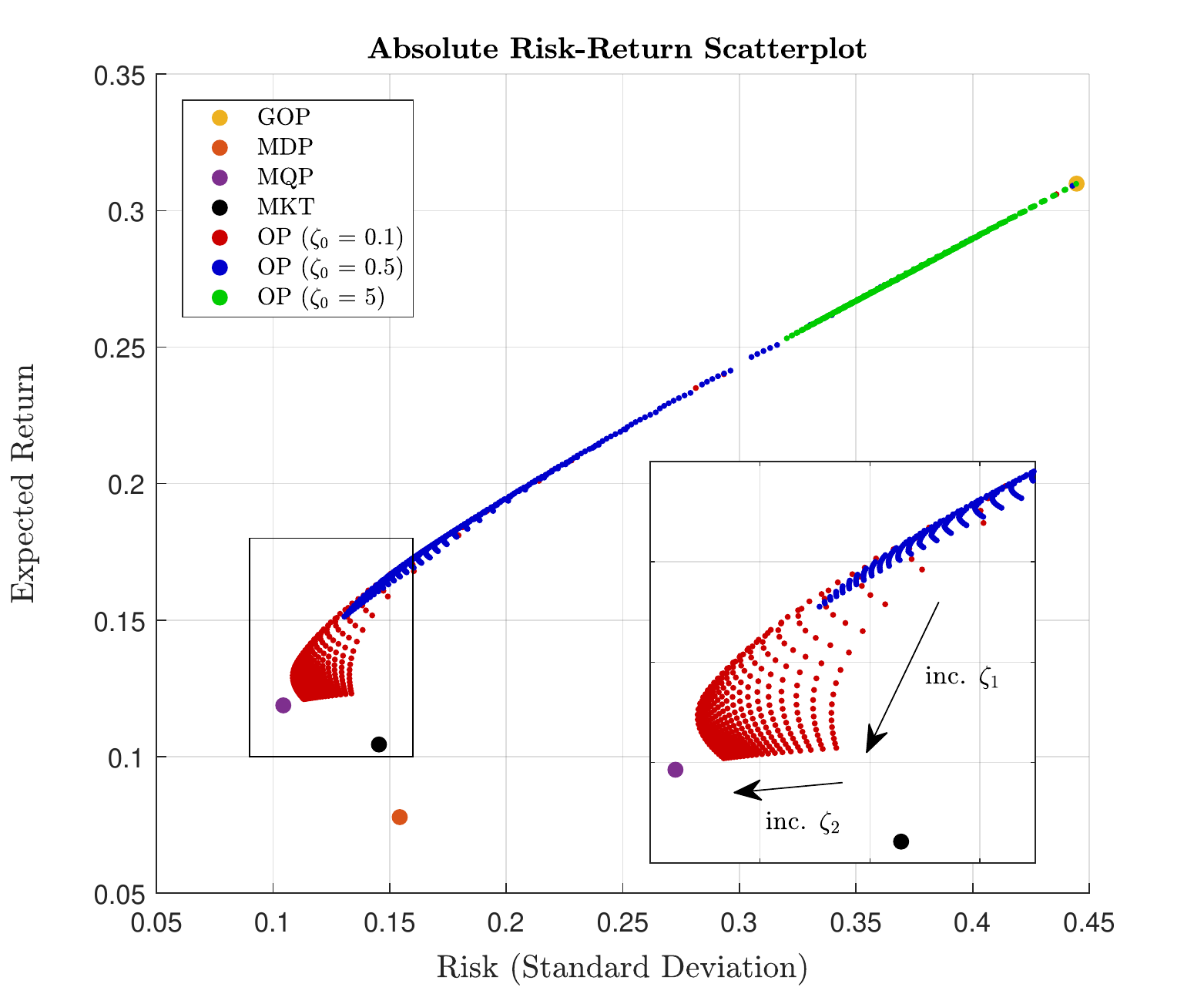}
	\includegraphics[width=0.48\textwidth]{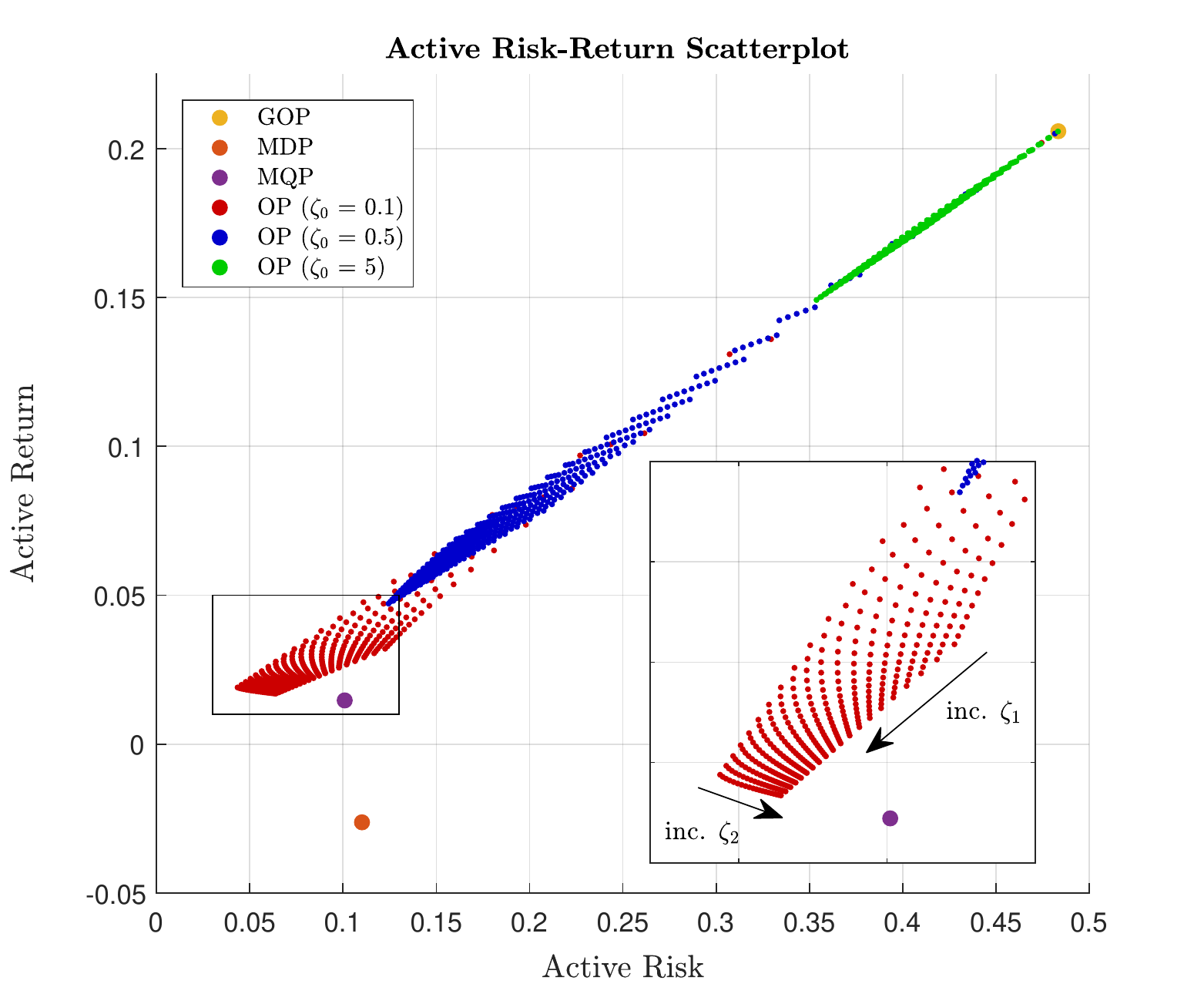}
	\caption{Scatterplots of absolute risk vs. absolute return (left panel) and active risk vs. active return (right panel) for the optimal portfolios with $\zeta_0 = 0.1$ (red),  $\zeta_0 = 0.5$ (blue),  $\zeta_0 = 5$ (green) and varying $\zeta_1, \zeta_2$ and GOP, MDP, MQP and market portfolio. Each dot is the average active risk/return for the corresponding portfolio across all simulations.}
	\label{fig:riskReturnScatter}
\end{figure}

The right panel in Figure \ref{fig:sharpeAndInfoRatioByZeta} shows the effect of varying $\zeta_2$ on the information ratio of the optimal strategy (note that the information ratio of the optimal portfolio is not sensitive to $\zeta_1$). As $\zeta_2 \rightarrow 0$, the information ratio of the optimal strategy converges to that of the GOP; as $\zeta_2 \rightarrow \infty$, the information ratio of the optimal strategy decreases to that of the MQP. Note that if the information ratio of the GOP was less than that of the MQP, the information ratio of the optimal strategy would \textit{increase} up to the MQP metric as $\zeta_2$ increases. The left panel in Figure \ref{fig:sharpeAndInfoRatioByZeta} shows the Sharpe ratio of the optimal portfolio for different $\vzeta$ values. The effect of varying these parameters on Sharpe ratio is more complex.

Finally, Figure \ref{fig:performanceDistribution} shows historgrams of the Sharpe ratio and information ratio of the optimal portfolio for $\vzeta = (0.5,0.5,0.5)$ compared to the MDP, GOP, MQP and the market portfolio. This gives a sense of the level of dispersion of these metrics for all portfolios.

\begin{figure}[h!]
\centering
	\includegraphics[width=0.48\textwidth]{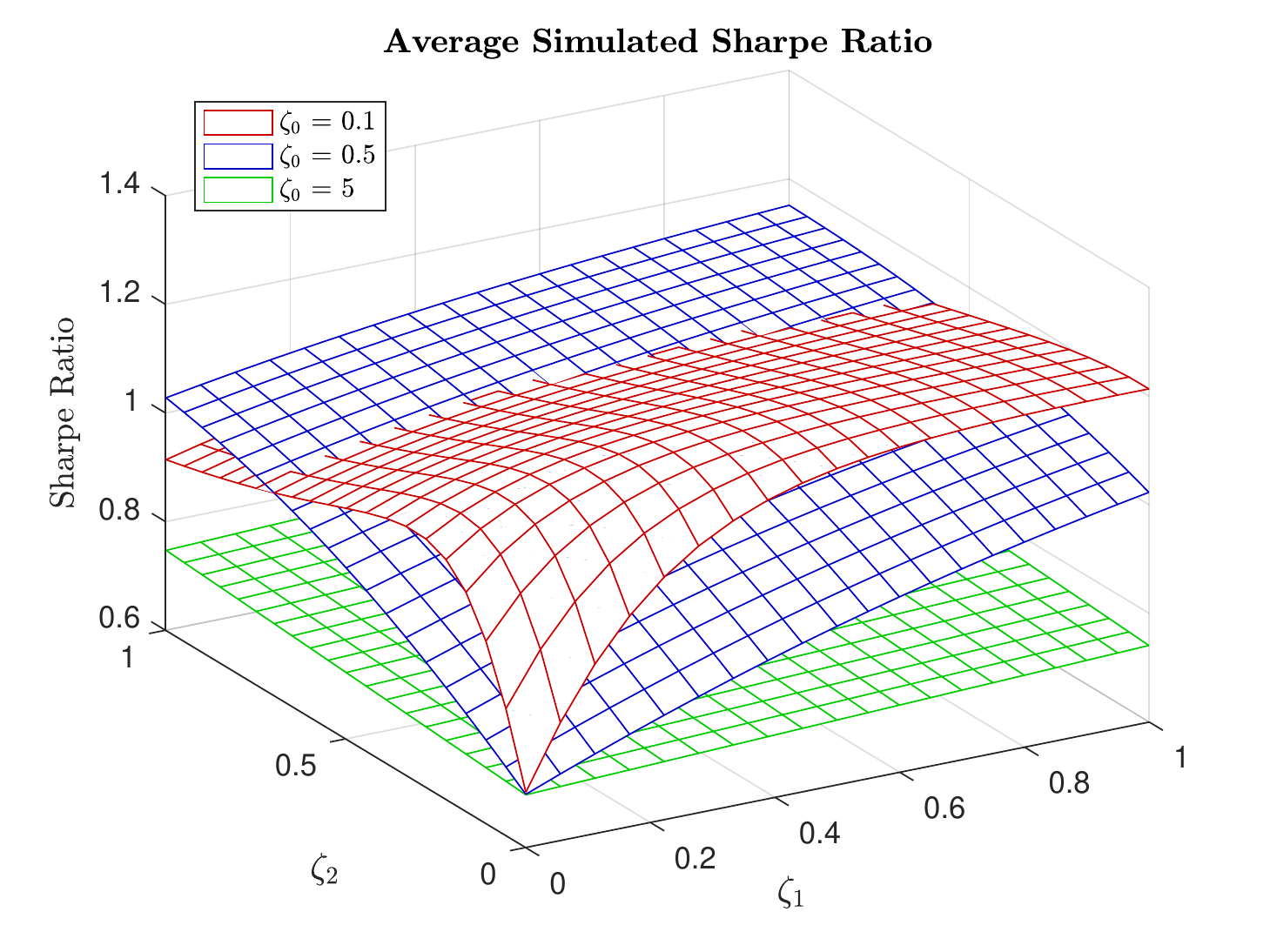}
	\includegraphics[width=0.48\textwidth]{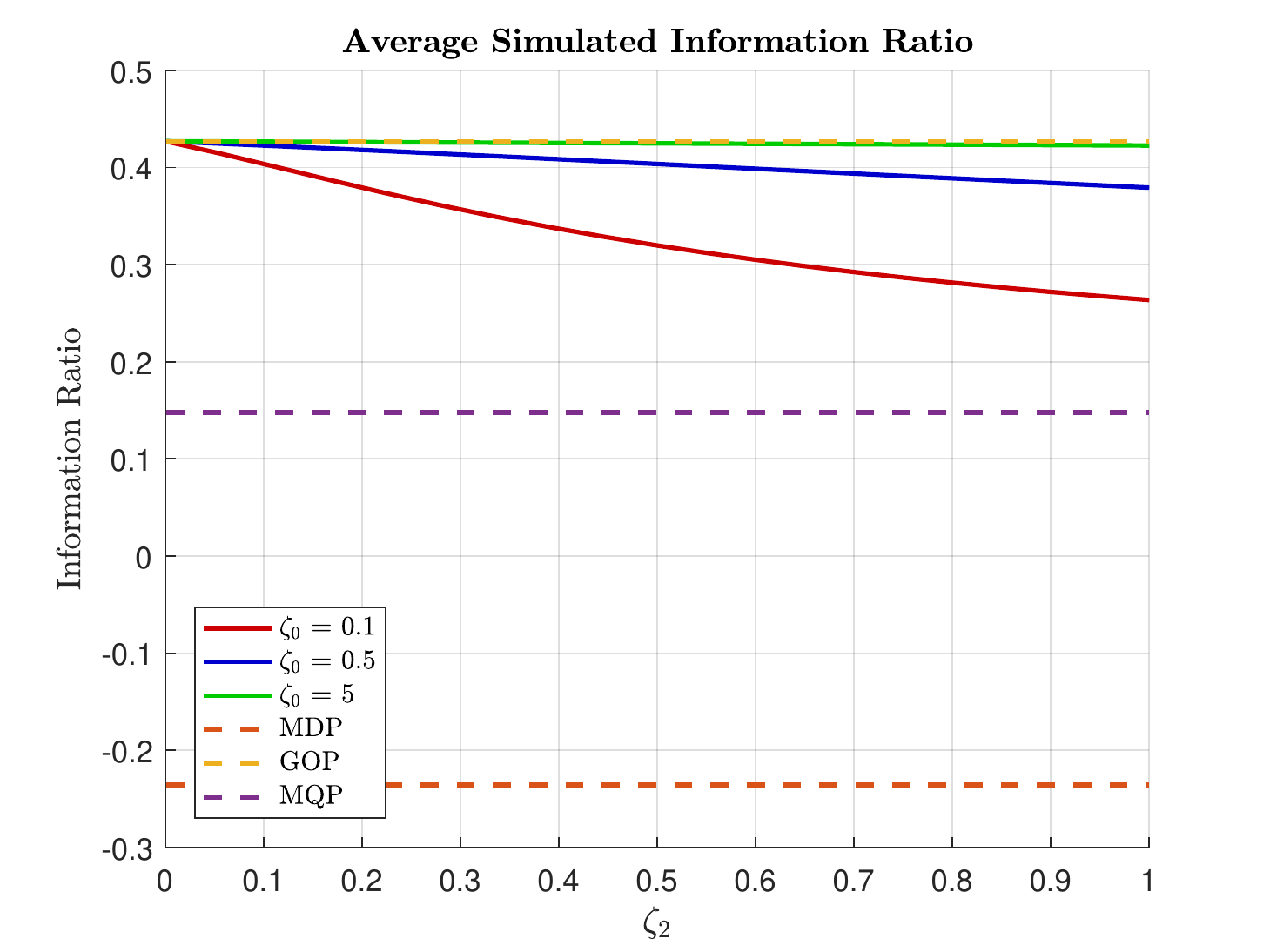}
	\caption{Left panel: Surfaces of average simulated Sharpe ratios for all $\vzeta$ values; each point on the surface is the average Sharpe ratio across all simulations for the optimal portfolio with the corresponding $\vzeta$. Right panel: average information ratio across simulations as a function of $\zeta_2$ compared to MDP, GOP, MQP.}
	\label{fig:sharpeAndInfoRatioByZeta}
\end{figure}

\clearpage

\begin{figure}[h!]
\centering
	\includegraphics[width=0.48\textwidth]{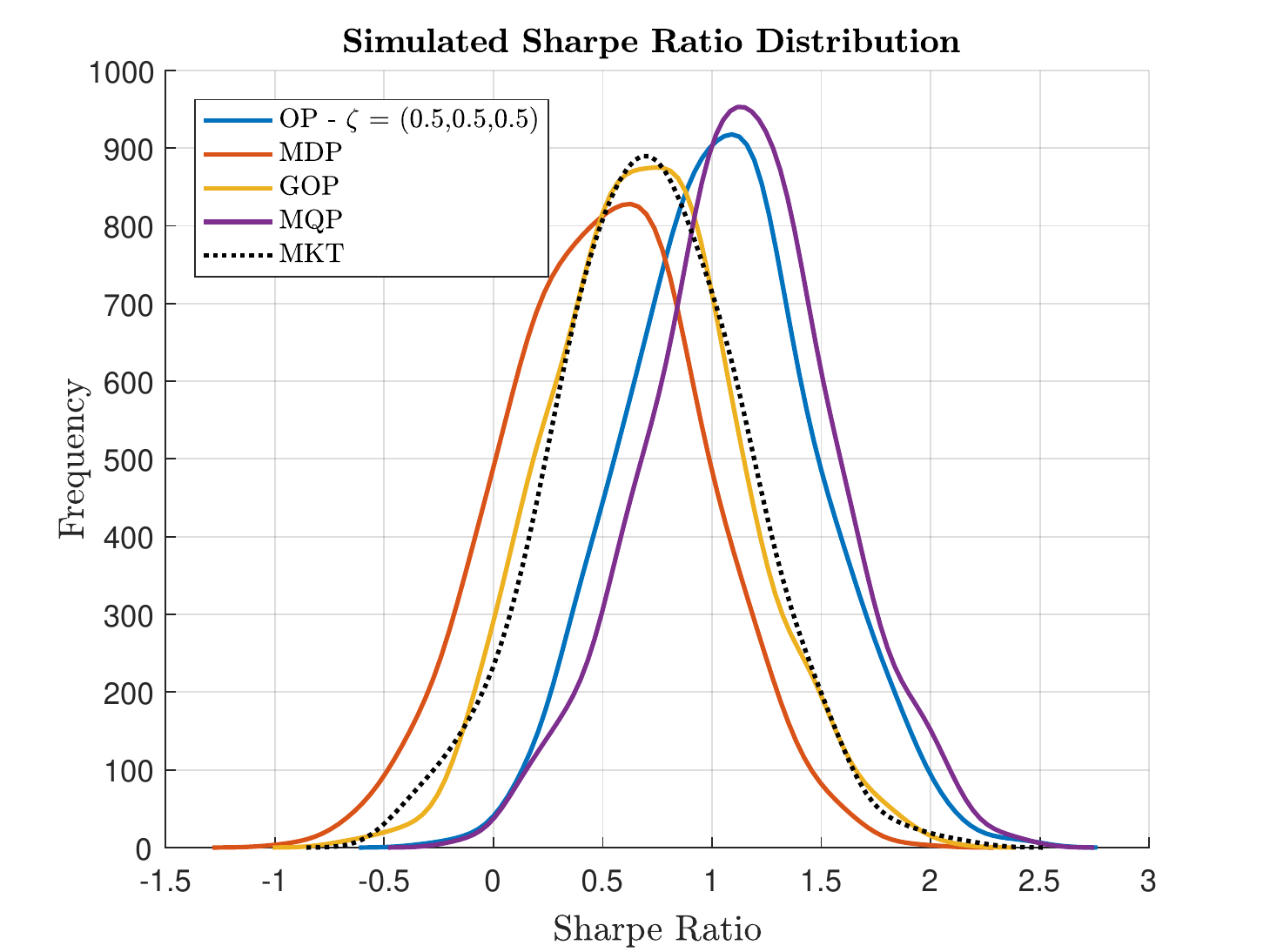}
		\includegraphics[width=0.48\textwidth]{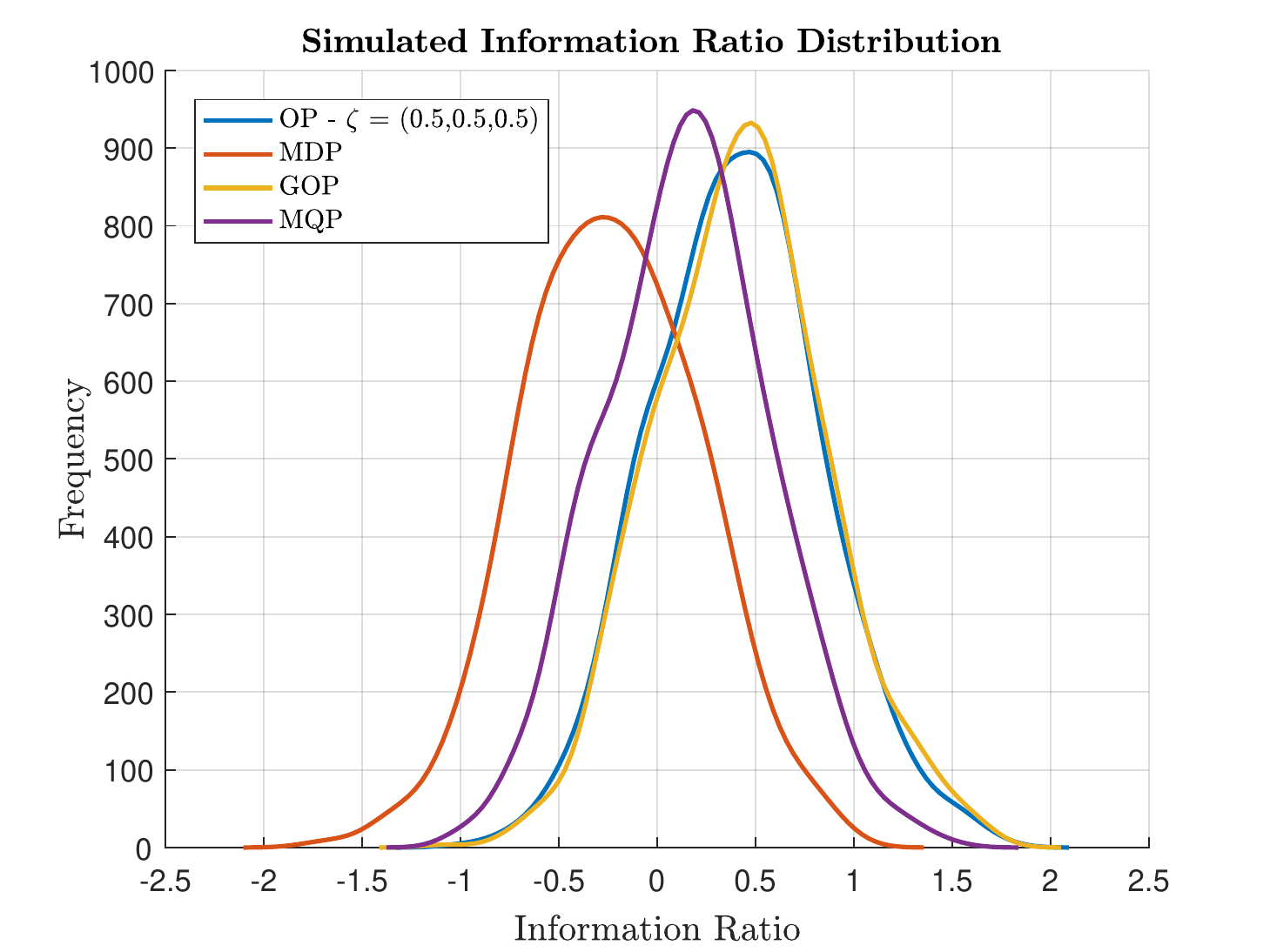}
	\caption{Distribution of simulated Sharpe ratios (left panel) and information ratios (right panel) for the optimal portfolio with $\vzeta = (0.5,0.5,0.5)$, MDP, GOP, MQP and the market portfolio.}
	\label{fig:performanceDistribution}
\end{figure}

\section{Conclusions}

In this article, we addressed the problem of how to dynamically allocate funds to beat a performance benchmark while remaining tethered to a tracking benchmark. An auxiliary penalty term is included to achieve additional goals, including controlling the quadratic variation of the wealth path and fine-tuning the allocation to individual securities. We related the solution to the growth optimal portfolio and minimum quadratic variance portfolio and discussed some of its limiting properties with respect to these portfolios. We furthermore provide a verification theorem demonstrating that our solution solves the optimal investment problem. Based on simulations of the market model, we show that the investor is able to control their risk-reward profile by tuning various tolerance parameters. Moreover, the optimal portfolios we find appear to outperform the static optimization of \cite{oderda2015} for all values of the penalty parameters in our simulation setting.

There are several future directions left open for research. One of these directions for which we already have preliminary results is related to incorporating latent information in the return and dividend process. Returns are notoriously difficult to estimate, and having a model which allows them to be stochastic, but also driven by latent factors, is essential to making the strategy robust to differing market environments. Extending the results to rank-based market models (e.g. Atlas models) that are prevalent in SPT literature is another interesting direction.

\begin{appendices}

\small

\renewcommand{\baselinestretch}{1}

\section{Proofs}
\subsection{Proof of Proposition \ref{prop:valueFn}} \label{proof:valueFn}
The PDE in \eqref{eqn:HJB} can be written as:
\begin{equation}
\begin{cases}
\left(\partial_t + \LL^\vX \right) h + \underset{\vpi \in \As}{\sup} \left\{  F(t,\vx,\vpi) \right\} = 0\;,\\
\hspace{4cm} h(T,\vx) = 0\;.
\end{cases}
\end{equation}
where, using the definition of $a(t,\vrho,\vpi)$ given in Proposition \ref{prop:Y_SDE} along with Assumption \ref{asmp:markov}, we have:
\begin{align*}
F(t,\vx,\vpi) &= \zeta_0 \cdot a(t,\rho(t,\vx),\vpi) - \tfrac{\zeta_1}{2} \cdot (\vpi - \eta(t,\vx))' \mSigma(t) (\vpi - \eta(t,\vx)) - \tfrac{\zeta_2}{2} \cdot \vpi' \mQ(t) \vpi \\
&= \zeta_0 \cdot \left[  \vpi' \valpha(t) - \tfrac{1}{2} \vpi' \mSigma(t) \vpi - \left(\vgamma_{\rho(t,\vx)}(t) + \vdelta_{\rho(t,\vx)}(t) \right) \right] \\
& \hspace{1cm} - \tfrac{\zeta_1}{2} \cdot \left[ \vpi' \mSigma(t) \vpi - 2 \vpi' \mSigma(t) \eta(t,\vx) + \eta(t,\vx)' \mSigma(t) \eta(t,\vx)  \right] - \tfrac{\zeta_2}{2} \cdot \vpi' \mQ(t) \vpi \\
& = - \tfrac{1}{2} \vpi' \left[ (\zeta_0 + \zeta_1) \mSigma(t) + \zeta_2 \mQ(t) \right] \vpi \\
& \hspace{1cm} + \vpi' \left[ \zeta_0 \valpha(t) + \zeta_1 \mSigma(t) \eta(t,\vx) \right] - \zeta_0 \cdot \left(\vgamma_{\rho(t,\vx)}(t) + \vdelta_{\rho(t,\vx)}(t) \right) - \tfrac{\zeta_1}{2} \cdot \eta(t,\vx)' \mSigma(t) \eta(t,\vx)  \\
& = -\tfrac{1}{2} \vpi' \mA(t) \vpi + \vpi' \vB(t,\vx) - C(t,\vx)
\end{align*}
which allows us to then rewrite the PDE above in the following concise form:
\begin{equation}
\begin{cases}
\left(\partial_t + \LL^\vX \right) h - C(t,\vx) + \underset{\vpi \in \As}{\sup} \left\{  -\tfrac{1}{2} \vpi' \mA(t) \vpi + \vpi' \vB(t,\vx) \right\} = 0\;,\\
\hspace{8.1cm} h(T,\vx) = 0\;.
\end{cases}
\end{equation}
where
\begin{align*}
\underset{\color{red} (n \times n)}{\mA(t)} &=  (\zeta_0 + \zeta_1) \mSigma(t) + \zeta_2 \mQ(t) \;,
 \\
\underset{\color{red} (n \times 1)}{\vB(t,\vx)} &= \zeta_0 \valpha(t) + \zeta_1 \mSigma(t) \eta(t,\vx) \;,
\quad \text{and}
\\
\underset{\color{red} (1 \times 1)}{C(t,\vx)} &= \zeta_0 \cdot \left(\vgamma_{\rho(t,\vx)}(t) + \vdelta_{\rho(t,\vx)}(t) \right) + \tfrac{\zeta_1}{2} \cdot \eta(t,\vx)' \mSigma(t) \eta(t,\vx) \;.
\end{align*}

The next step is to solve for the optimal control by determining the first order conditions. The optimization problem of interest in this case is:
\begin{equation*}
\underset{\vpi}{\max} ~  \left\{-\tfrac{1}{2} \vpi' \mA(t) \vpi + \vpi' \vB(t,\vx)  \right\} \quad \text{subject to } \quad \ones' \vpi = 1\;.
\end{equation*}
This can be solved using a Lagrange multiplier to find that the optimal control is
\begin{equation} \label{FOC}
\vpi_\vzeta^*(t) = \mA^{-1}(t) \cdot \left[ \frac{1 - \ones' \mA^{-1}(t) \vB(t,\vx)}{\ones' \mA^{-1}(t) \ones} \cdot \ones + \vB(t,\vx) \right] = \mA^{-1}(t) \cdot \vd(t,\vx)\;,
\end{equation}
where $\underset{\color{red} (n \times 1)}{\vd(t,\vx) }= \frac{1 - \ones' \mA^{-1}(t) \vB(t,\vx)}{\ones' \mA^{-1}(t) \ones} \cdot \ones + \vB(t,\vx) $. Note that $\mA$ is positive definite, so the optimal control is indeed a maximizer. Next, notice that we can write the quadratic term in the PDE above (suppressing the dependence on $t$ and $\vx$ for notational simplicity) as follows:
\begin{align}
-\tfrac{1}{2} \left( \vpi_\vzeta^* \right)' \mA \vpi_\vzeta^* + \left( \vpi_\vzeta^* \right)' \vB &= -\tfrac{1}{2} \vd' \mA^{-1} \cdot \mA \cdot \mA^{-1} \cdot \vd + \vd' \mA^{-1} \vB \quad \quad \quad \quad \mbox{\footnotesize (note that $\mA$ and $\mA^{-1}$ are symmetric)} \nonumber \\
&= \vd' \mA^{-1} \left[ -\tfrac{1}{2} \vd + \vB \right]  \nonumber \\
&= -\tfrac{1}{2} \left[ \frac{1 - \ones' \mA^{-1} \vB}{\ones' \mA^{-1} \ones} \cdot \ones + \vB  \right]' \mA^{-1} \left[  \frac{1 - \ones' \mA^{-1} \vB}{\ones' \mA^{-1} \ones} \cdot \ones - \vB \right]
\end{align}
Therefore, the PDE in \eqref{eqn:HJB} becomes:
\begin{equation} \label{eqn:HJB2}
\begin{cases}
\left( \partial_t + \LL^\vX \right) h - G(t,\vx) = 0\;,\\
\hspace{2.45cm} h(T,\vx) = 0\;,
\end{cases}
\end{equation}
where
\begin{align*}
G(t,\vx) =  C(t,\vx) + \frac{1}{2} \left[ \frac{1 - \ones' \mA^{-1}(t) \vB(t,\vx)}{\ones' \mA^{-1}(t) \ones} \cdot \ones + \vB(t,\vx)  \right]' \mA^{-1}(t) \left[  \frac{1 - \ones' \mA^{-1}(t) \vB(t,\vx)}{\ones' \mA^{-1}(t) \ones} \cdot \ones - \vB(t,\vx) \right] \;.
\end{align*}

The remainder of the proof relies on finding a Feynman-Kac representation for the solution to the PDE above. It is slightly more convenient at this point to perform a change of variables and use $\log \vX$ as the state variables. Note that this mainly affects the infinitesimal generator in the PDE above. Namely, we have:
\[ \LL^\vX = \tfrac{1}{2} \sum_{i,j=1}^{n} { \color{red} \underbrace{\color{black} \mSigma_{ij}(t)}_{= ~a_{ij}} } \cdot \partial_{ij} + \sum_{i=1}^{n} {\color{red} \underbrace{ \color{black} \left(\gamma_i(t) + \delta_i(t)\right)}_{= ~b_i}} \cdot \partial_i  \]
First, we establish the existence of a sufficiently smooth solution to the PDE \eqref{eqn:HJB2}. For this, we rely on the sufficient conditions given in Remark 7.8 in Chapter 5 of \cite{karatzas2012}, namely: (i) uniform ellipticity of the PDE on $[0,\infty) \times \RR^n$; (ii) boundedness of $a_{ij}, b_i$ in $[0,T] \times \RR^n$; (iii) uniform H\"{o}lder continuity of $a_{ij}, b_i$ and $G$ in $[0,T] \times \RR^n$; (iv) $G$ satisfies $|G(t,\vx)| \leq K(1+\|\vx\|^{2\lambda})$ for some $K > 0, \lambda \geq 1$.

The first condition holds by the nondegeneracy assumption we made on the covariance matrix $\mSigma$. For the second condition, we have assumed that $\vgamma$ and $\vdelta$ are bounded above and below for all $t$ and we can use the nondegeneracy and bounded variance assumptions made in Section \ref{sec:modelSetup} to obtain upper and lower bounds on the elements of $\mSigma(t)$. In particular, since we have that for any $\mathbf{x},\mathbf{y} \in \RR^n$ and any $t \geq 0$:
\begin{align*}
& \varepsilon_1 \| \mathbf{x} + \mathbf{y} \|^2 ~\leq~ (\mathbf{x} + \mathbf{y})' \mSigma(t) (\mathbf{x} + \mathbf{y}) ~\leq~ M_1 \| \mathbf{x} + \mathbf{y} \|^2 \\
\implies \quad & \varepsilon_1 \| \mathbf{x} + \mathbf{y} \|^2 ~\leq~ \mathbf{x}' \mSigma(t) \mathbf{x} + \mathbf{y}' \mSigma(t) \mathbf{y} + 2 \mathbf{x}' \mSigma(t) \mathbf{y}  ~\leq~ M_1 \| \mathbf{x} + \mathbf{y} \|^2 \\
\implies \quad & \tfrac{1}{2} \left( \varepsilon_1 \| \mathbf{x} + \mathbf{y} \|^2 - M_1 \left( \| \mathbf{x} \|^2 + \| \mathbf{y} \|^2 \right) \right) ~\leq~  \mathbf{x}' \mSigma(t) \mathbf{y}  ~\leq~ M_1 \| \mathbf{x} + \mathbf{y} \|^2
\end{align*}
Taking $\mathbf{x} = \mathbf{e}_i$ and $\mathbf{y} = \mathbf{e}_j$, where $\mathbf{e}_i$ is vector of zeros with a 1 in the $i^{\scriptsize \mbox{th}}$ position, we obtain the following bound on the covariance elements:
\[ \varepsilon_1 - M_1 ~\leq~  \mSigma_{ij}(t) ~\leq~  4M_1 \;. \]
This shows that condition (ii) is satisfied. Next, we show that the function $G$ is bounded and hence satisfies the polynomial growth condition (iv), which will also be needed to prove that condition (iii) holds. For this, we note that since the elements of $\mSigma$ along with $\rho(t,\vx)$ and $\eta(t,\vx)$ are bounded, we conclude that the functions $\vB(t,\vx)$ and $C(t,\vx)$ are bounded. Now, it suffices to show that quadratic forms $\mathbf{x}' \mA^{-1}(t) \mathbf{y}$ for $\mathbf{x}, \mathbf{y}$ in some bounded subset of $\RR^n$ are also bounded. To this end, recall that $\mA(t) =  (\zeta_0 + \zeta_1) \mSigma(t) + \zeta_2 \mQ(t)$ and that for all $\mathbf{x} \in \RR^n$:
\begin{align*}
\varepsilon_1 \| \mathbf{x} \|^2 ~\leq~ \mathbf{x}' \mSigma(t) \mathbf{x} ~\leq~ M_1 \| \mathbf{x} \|^2  \\
\varepsilon_2 \| \mathbf{x} \|^2 ~\leq~ \mathbf{x}' \mQ(t) \mathbf{x} ~\leq~ M_2 \| \mathbf{x} \|^2
\end{align*}
which implies that there exists $\epsilon, M > 0$ such that for all $\mathbf{x} \in \RR^n$:
\[ \qquad \epsilon \| \mathbf{x} \|^2 ~\leq~ \mathbf{x}' \mA(t) \mathbf{x} ~\leq~ M \| \mathbf{x} \|^2 \]
Denoting the eignevalues of $\mA(t)$ by $\lambda_1(t) \geq ... \geq \lambda_n(t)$, we note that the bounds above are tight when we replace $M$ and $\varepsilon$ with $\lambda_1$ and $\lambda_n$, respectively. Thus, the eigenvalues satisfy:
\[ \varepsilon \leq \lambda_n(t) \leq \lambda_1(t) \leq M \qquad \implies \qquad \frac{1}{M} \leq \frac{1}{\lambda_1(t)} \leq \frac{1}{\lambda_n(t)} \leq \frac{1}{\varepsilon} \]
from which we can conclude there exist positive constants $\tilde{\varepsilon}, \tilde{M}$ such that for any $\mathbf{x} \in \RR^n$:
\[ \tilde{\varepsilon} \| \mathbf{x} \|^2 ~\leq~ \mathbf{x}' \mA^{-1}(t) \mathbf{x} ~\leq~  \tilde{M} \| \mathbf{x} \|^2 \]
and by similar reasoning as above we can obtain:
\[ \tfrac{1}{2} \left( \tilde{\varepsilon} \| \mathbf{x} + \mathbf{y} \|^2 - \tilde{M} \left( \| \mathbf{x} \|^2 + \| \mathbf{y} \|^2 \right) \right) ~\leq~  \mathbf{x}' \mA^{-1}(t) \mathbf{y}  ~\leq~ \tilde{M} \| \mathbf{x} + \mathbf{y} \|^2 \]
for all $\mathbf{x}, \mathbf{y} \in \RR^n$ . When $\mathbf{x}, \mathbf{y}$ are in a bounded subset of $\RR^n$, the norms on either side of the inequality above are also bounded. Combining this with the fact that $\vB$ and $C$ are bounded we can conclude that $G$ is in fact bounded (since we are concerned with quadratic forms involving $\ones$ and $\vB$), and hence satisfies the polynomial growth condition (iv).

To prove condition (iii) we note that the assumption of differentiability on $\gamma_i$, $\delta_i$ and $\xi_{i \nu}$ implies that $a_{ij} = \mSigma_{ij} = \sum_{\nu = 1}^{k} \xi_{i \nu} \xi_{j \nu}$ and $b_i = \gamma_i + \delta_i$ are differentiable and hence uniformly H\"{o}lder continuous. Similarly, $G$ can be shown to be differentiable and hence uniformly H\"{o}lder continuous on $[0,T] \times \RR^n$ by noting that $\mA$, $\vB$ and C are all bounded and differentiable which implies the same for $G$. Since conditions (i)-(iv) are satisfied, a solution exists to the PDE \eqref{eqn:HJB2} and, by Theorem 7.6 in Chapter 5 of \cite{karatzas2012}, the unique solution is given by the following Feynman-Kac representation:
\begin{equation}
h(t,\vx)= \EE_{t,\vx} \left[ -\int_t^T G(s,\vX(s)) ~ ds \right],
\end{equation}
where the expectation is taken under the physical measure $\PP$.
\hfill
$\blacksquare$

\subsection{Proof of Theorem \ref{thm:optCont}} \label{proof:optCont}

Along the lines of Theorem 5.1 of \cite{touzi2012}, we provide a verification argument to demonstrate that the candidate solution given in Proposition \ref{prop:valueFn} is in fact the value function, and that the corresponding control given in \eqref{eqn:opt-pi-A-and-B} is the optimal control. Since $h(t,\vx)$ in (\ref{valueFn}) is a classical solution to the HJB equation, and $\vpi_\vzeta^*(t)$ is the maximizer of
\[ \vpi ~\longmapsto~ F(t,\vx,\vpi)\;, \]
it suffices to check that the process $\vpi_\vzeta^* = \left(\vpi_\vzeta^*(t)\right)_{t \geq 0}$ is a well-defined admissible control process and that the controlled stochastic differential equation
\begin{align*}
dY^{\vpi_\vzeta^*,\vrho}(t) ~=~ a\left( t,\vrho(t),\vpi_\vzeta^*(t) \right) ~dt  + \left( \vxi_{\vpi_\vzeta^*} (t) - \vxi_\vrho(t) \right)' ~ d\vW(t)\;,
\end{align*}
defines a unique solution for each given initial data.

First, we verify that $\vpi^*_\vzeta$ is an admissible control. This requires us to show that $\vpi^*_\vzeta$ is an $\mathfrak{F}$-adapted vector-valued process with $\vpi^*_\vzeta(t)' \ones = 1$ for all $t$ and that $\pi^*_i(t)$ is bounded almost surely for all $t$ and for $i = 1,..., n$. Clearly, $\vpi^*_\vzeta$ is $\mathfrak{F}$-adapted and  its elements sum up to 1 from the following observation:
\begin{align*}
\ones' \vpi^*_\vzeta ~&=~ \ones' \mA^{-1} \cdot \left[ \frac{1 - \ones' \mA^{-1} \vB}{\ones' \mA^{-1} \ones} \cdot \ones + \vB \right] \\
~&=~ \frac{1 - \ones' \mA^{-1} \vB}{\ones' \mA^{-1} \ones} \cdot \ones' \mA^{-1} \ones + \ones' \mA^{-1} \vB \\
~&=~ 1 - \ones' \mA^{-1} \vB + \ones' \mA^{-1} \vB \\
~&=~ 1\;.
\end{align*}
The fact that $\vpi^*_\vzeta$ is bounded follows from the boundedness of $\vB$ and the elements of $\mA^{-1}$ which was established in the previous proof. Next, we need to show the existence and uniqueness of the solution to the SDE given in the outset of this proof. To this end, note that the drift and volatility terms for any portfolio $\vpi$ are $a(t,\vrho(t),\vpi(t)) = \vpi' \valpha(t) - \tfrac{1}{2} \vpi' \mSigma(t) \vpi - \left(\gamma_\vrho(t) + \delta_\vrho(t) \right)$ and $\left( \vxi_\vpi (t) - \vxi_\vrho(t) \right)$ respectively. Since any admissible $\vpi$, $\vrho$, $\vgamma,$ $\vdelta$ and the elements of $\mSigma(t)$ are all bounded, so too is the drift term $a(t,\vrho(t),\vpi(t))$. In particular it is square integrable for any $t \geq 0$. Moreover, since $\vxi(t)$ is assumed to be square integrable, so is the volatility term. Finally, since neither the drift nor volatility depends on $Y^{\vpi,\vrho}$, both satisfy the required Lipschitz condition given in Theorem 2.2 of \cite{touzi2012}. Then, by the same theorem, the SDE admits a unique solution for any choice of initial data $Y^{\vpi,\vrho}(0)$. \hfill $\blacksquare$

\subsection{Proof of Corollary \ref{cor:properties}} \label{proof:properties}
We first prove (i). Recall that the optimal portfolio is given by:
\begin{align*}
\vpi^*_\vzeta = \mA^{-1} \left[ \frac{1 - \ones' \mA^{-1} \vB}{\ones' \mA^{-1} \ones} \cdot \ones + \vB \right] = \left( 1 - \ones' \mA^{-1} \vB \right) \cdot \frac{1}{\ones' \mA^{-1} \ones} \mA^{-1} \ones + \mA^{-1} \vB
\end{align*}
where $\mA= (\zeta_0 + \zeta_1) \mSigma + \zeta_2 \mQ$ and $\vB = \zeta_0 \valpha + \zeta_1 \mSigma \eta$ and the dependence on $t$ and $\vx$ is dropped for brevity. Now, we can write $\mA$ in the following two ways:
\begin{align*}
\mA &= (\zeta_0 + \zeta_1) \left( \mSigma + \frac{\zeta_2}{\zeta_0 + \zeta_1} \mQ \right)
&&	\mA = \zeta_2 \left( \frac{\zeta_0 + \zeta_1}{\zeta_2} \mSigma + \mQ \right) \\
&= (\zeta_0 + \zeta_1) \mA_1 && ~~~ = \zeta_2 \mA_2
\end{align*}
from which it follows that:
\begin{align*}
\underset{\zeta_0 \to \infty}{\lim} \mA_1 &= \underset{\zeta_1 \to \infty}{\lim} \mA_1 = \mSigma \qquad \text{and} \qquad \underset{\zeta_2 \to \infty}{\lim} \mA_2 = \mQ
\end{align*}
Additionally, we can write the inverse of $\mA$ as follows:
\[ \mA^{-1} = \frac{1}{\zeta_0 + \zeta_1} \mA^{-1}_1 = \frac{1}{\zeta_2} \mA^{-1}_2 \]
Writing the optimal portfolio in terms of $\mA_1$ we have:
\begin{align*}
\vpi^*_\vzeta &= \left[ 1 - \frac{1}{\zeta_0 + \zeta_1} \ones' \mA_1^{-1} \left( \zeta_0 \valpha +  \zeta_1 \mSigma \eta \right) \right] \cdot \frac{1}{\ones' \mA_1^{-1} \ones} \mA_1^{-1} \ones + \frac{1}{\zeta_0 + \zeta_1} \mA_1^{-1}  \left( \zeta_0 \valpha + \zeta_1 \mSigma \eta \right) \\
&= \left[ 1 - \frac{\zeta_0}{\zeta_0 + \zeta_1} \ones' \mA_1^{-1} \valpha - \frac{\zeta_1}{\zeta_0 + \zeta_1} \ones' \mA_1^{-1} \mSigma \eta \right] \cdot \frac{1}{\ones' \mA_1^{-1} \ones} \mA_1^{-1} \ones + \frac{\zeta_0}{\zeta_0 + \zeta_1} \mA_1^{-1} \valpha + \frac{\zeta_1}{\zeta_0 + \zeta_1} \mA_1^{-1} \mSigma \eta
\end{align*}
Now, noting that:
\begin{align*}
\underset{\zeta_0 \to \infty}{\lim} ~\frac{\zeta_0}{\zeta_0 + \zeta_1} = \underset{\zeta_1 \to \infty}{\lim} ~\frac{\zeta_1}{\zeta_0 + \zeta_1} = 1 \\
\underset{\zeta_0 \to \infty}{\lim} ~\frac{\zeta_1}{\zeta_0 + \zeta_1} = \underset{\zeta_1 \to \infty}{\lim} ~\frac{\zeta_0}{\zeta_0 + \zeta_1} = 0
\end{align*}
it follows that:
\[ \underset{\zeta_0 \to \infty}{\lim} \vpi^*_\vzeta = \vpi_{GOP}  \qquad , \qquad
   \underset{\zeta_1 \to \infty}{\lim} \vpi^*_\vzeta = \eta \;. \]
Similarly, writing the optimal portfolio in terms of $\mA_2$ we have:
\begin{align*}
\vpi^*_\vzeta &= \left[ 1 - \frac{\zeta_0}{\zeta_2} \ones' \mA_2^{-1} \valpha - \frac{\zeta_1}{\zeta_2} \ones' \mA_2^{-1} \mSigma \eta \right] \cdot \frac{1}{\ones' \mA_2^{-1} \ones} \mA_2^{-1} \ones + \frac{\zeta_0}{\zeta_2} \mA_2^{-1} \valpha + \frac{\zeta_1}{\zeta_2} \mA_2^{-1} \mSigma \eta
\end{align*}
from which it follows that:
\[ \underset{\zeta_2 \to \infty}{\lim} \vpi^*_\vzeta = \frac{1}{\ones' \mQ^{-1} \ones} \mQ^{-1} \ones \;. \]
Statements (ii) and (iii) can be verified directly by substituting in the appropriate values of $\zeta_0, \zeta_1, \zeta_2$ and $\mQ$ and noting that the GOP is given by $\vpi_{GOP} = (1 - \ones' \mSigma^{-1} \valpha) \cdot \vpi_{MQP} + \mSigma^{-1} \valpha$. \hfill $\blacksquare$

\end{appendices}

\small
\section*{References}
\bibliographystyle{chicago}
\bibliography{OutperformanceAndTracking}

\end{document}